%% file: main.tex
\documentclass[11pt]{amsart}
\usepackage[top=1.0in, bottom=1.0in, left=1.2in, right=1.2in]{geometry}

\usepackage{amsfonts, amssymb, amsmath, enumerate}
\usepackage{hyperref}
\usepackage[foot]{amsaddr}
\usepackage{bbm}
\usepackage{xcolor}
\usepackage[style=nature]{biblatex}
\usepackage{wasysym}
\hypersetup{
	colorlinks,
	linkcolor={red!40!gray},
	citecolor={blue!40!gray},
	urlcolor={blue!70!gray}
}

\numberwithin{equation}{section}
\usepackage{amsfonts, amssymb, amsmath, enumerate}

\usepackage{mathrsfs}

\usepackage{mathtools}
\usepackage[normalem]{ulem}
\usepackage{comment}

\usepackage{bbm}
\usepackage{algorithm}
\usepackage{algorithmic}

\numberwithin{equation}{section}

\DeclareUnicodeCharacter{2113}{l}

\definecolor{darkgreen}{cmyk}{0.6,0,0.8,0}
\newcommand{\Green}[1]{{\color{darkgreen}  {#1}}}
\def\cmark{\Green{\checkmark}}

\DeclareMathOperator{\E}{\mathbb{E}}
\DeclareMathOperator{\R}{\mathbb{R}}

\DeclareMathOperator{\Pb}{\mathbb{P}}
\DeclareMathOperator{\Var}{Var}

\DeclareMathOperator{\card}{card}

\DeclareMathOperator{\polylog}{polylog}
\DeclareMathOperator{\poly}{poly}
\DeclareMathOperator{\argmin}{argmin}
\DeclareMathOperator{\nnz}{nnz}

\DeclareMathOperator{\augsym}{augsym}
\DeclareMathOperator{\aug}{aug}
\DeclareMathOperator{\Id}{Id}
\DeclareMathOperator{\cov}{Cov}

\DeclareMathOperator{\spec}{spec}

\newcommand{\cN}{\mathcal{N}}

\DeclarePairedDelimiter{\norm}{\lVert}{\rVert}
\DeclarePairedDelimiter{\ip}{\langle}{\rangle}

\newtheorem{theorem}{Theorem}[section]

\newtheorem{lemma}[theorem]{Lemma}

\theoremstyle{remark}
\newtheorem{remark}[theorem]{Remark}

\theoremstyle{definition}
\newtheorem{definition}[theorem]{Definition}

\newtheorem*{def:OSE}{Definition \ref{def:OSE}}

\begin{document}

\title[Optimal Embedding Dimension for Sparse Subspace Embeddings]{
Optimal Embedding Dimension for Sparse\\Subspace Embeddings
}
\author{Shabarish Chenakkod}
\author{Micha{\l} Derezi\'nski}
\author{Xiaoyu Dong}
\author{Mark Rudelson}
 \thanks{Partially supported by DMS 2054408.}
\address{University of Michigan, Ann Arbor, MI}
 \email{shabari@umich.edu, derezin@umich.edu, xydong@umich.edu, rudelson@umich.edu}

 \begin{abstract}
  A random $m\times n$ matrix $S$ is an oblivious subspace embedding (OSE) with parameters $\epsilon>0$, $\delta\in(0,1/3)$ and $d\leq m\leq n$, if
  for any $d$-dimensional subspace $W\subseteq \R^n$,
  $$\Pb\big(\,\forall_{x\in W}\ (1+\epsilon)^{-1}\|x\|\leq
  \|Sx\|\leq (1+\epsilon)\|x\|\,\big)\geq 1-\delta.$$
 It is known that the embedding dimension of an OSE must satisfy $m\geq d$, and for
 any
 $\theta > 0$, a Gaussian embedding matrix with $m\geq (1+\theta) d$ is an OSE with $\epsilon = O_\theta(1)$. However, such optimal embedding dimension is not known for other embeddings. Of particular interest are sparse OSEs, having $s\ll m$ non-zeros per column (Clarkson and Woodruff, STOC 2013), with applications to problems such as least squares regression and low-rank approximation.

 We show that, given any $\theta > 0$, an $m\times n$ random matrix $S$ with $m\geq (1+\theta)d$ consisting of randomly sparsified $\pm1/\sqrt s$ entries and having $s= O(\log^4(d))$ non-zeros per column, is an oblivious subspace embedding with $\epsilon = O_{\theta}(1)$.
 Our result addresses the main open question posed by Nelson and Nguyen (FOCS 2013), who conjectured that sparse OSEs can achieve $m=O(d)$ embedding dimension, and it improves on $m=O(d\log(d))$ shown by Cohen (SODA 2016).
 We use this to construct the first oblivious subspace embedding with $O(d)$ embedding dimension that can be applied faster than current matrix multiplication time, and to obtain an optimal single-pass algorithm for least squares regression.
 We further extend our results to Leverage Score Sparsification (LESS), which is a recently introduced non-oblivious embedding technique. We use LESS to construct the first subspace embedding with low distortion $\epsilon=o(1)$ and optimal embedding dimension $m=O(d/\epsilon^2)$ that can be applied in current matrix multiplication time, addressing a question posed by Cherapanamjeri, Silwal, Woodruff and Zhou (SODA 2023).

Our analysis builds on recent advances in universality theory for random matrices, to overcome the limitations of existing approaches such as those based on matrix concentration inequalities. In the process, we establish new bounds on the extreme singular values for a class of nearly-square random matrices that arise from applying an $m\times n$ sparse OSE matrix to an $n\times d$ isometric embedding matrix, which may be of independent interest. To maximally leverage these results, we introduce new non-uniformly sparsified embedding constructions which, in particular, reduce the random bit complexity of our OSE to $\polylog(n)$.
\end{abstract}
 
\maketitle

\newpage

\input{intro.tex}

\input{preliminaries.tex}

\input{oblivious.tex}

\input{nonoblivious.tex}

\input{applications.tex}

\printbibliography
\end{document}

%% file: intro.tex
\section{Introduction}
\label{s:intro}
Since their first introduction by Sarl\'os \cite{sarlos2006improved}, subspace embeddings have been used as a key dimensionality reduction technique in designing fast approximate randomized algorithms for numerical linear algebra problems, including least squares regression, $l_p$ regression, low-rank approximation, and approximating leverage scores, among others \cite{rokhlin2008fast,clarkson2009numerical,meng2013low,clarkson2013low,clarkson2016fast,drineas2012fast,wang2022tight}; we refer to the surveys \cite{halko2011finding,woodruff2014sketching,drineas2016randnla,martinsson2020randomized,derezinski2021determinantal,randlapack_book} for an overview. The subspace embedding property can be viewed as an extension of the classical Johnson-Lindenstrauss (JL) lemma \cite{johnson1984extensions}, which provides a transformation that reduces the dimensionality of a finite set of $n$-dimensional vectors while preserving their pairwise distances (i.e., an embedding). Here, instead of a finite set, we consider a $d$-dimensional subspace of $\R^n$, where $d\ll n$. One way to define such a transformation is via an $m\times n$ random matrix $S$, where $m\ll n$ is the \emph{embedding dimension}. Remarkably, when the embedding dimension $m$ and the distribution of $S$ are chosen correctly, then matrix $S$ can provide an embedding for any $d$-dimensional subspace $W$ with high probability. Such a distribution of $S$ is \emph{oblivious} to the choice of the subspace. Since any $d$-dimensional subspace of $\R^n$ can be represented as the range of an $n\times d$ matrix $U$ with orthonormal columns, we arrive at the following~definition.

\begin{definition}\label{def:OSE}
A random $m\times n$ matrix $S$ is an $(\epsilon,\delta)$-\textit{subspace embedding} (SE) for an $n\times d$ matrix $U$ with orthonormal columns, where $\epsilon>0$, $\delta\in(0,1/3)$ and $d\leq m\leq n$, if
  \begin{align*}
    \Pb\big(\,\forall_{x\in \R^d}\ 
    (1+\epsilon)^{-1}\|x\|\leq\|SUx\|\leq(1+\epsilon)\|x\|\,\big)\geq 1-\delta.
  \end{align*}
  If $S$ is an $(\epsilon,\delta)$-SE for all such $U$, then it is an $(\epsilon,\delta,d)$-\textit{oblivious subspace embedding} (OSE). 
\end{definition}

The embedding dimension of any OSE must satisfy $m\geq d$, so ideally we would like an embedding with $m$ as close to $d$ as possible, while making sure that the distortion $\epsilon$ is not too large. For many applications, a constant distortion, i.e., $\epsilon = O(1)$, is sufficient, while for some, one aims for low-distortion embeddings with $\epsilon\ll 1$. 

One of the most classical families of OSE distributions are Gaussian matrices $S$ with i.i.d.~entries $s_{i,j}\sim \cN(0,1/m)$. Thanks to classical results on the spectral distribution of Gaussian matrices \cite{rudelson2010non} combined with their rotational invariance, there is a sharp characterization of the distortion factor $\epsilon$ for Gaussian subspace embeddings, which implies that the embedding dimension can be arbitrarily close to $d$, i.e., $m=(1+\theta)d$ for any constant $\theta > 0$, while ensuring the OSE property with $\epsilon= O(1)$ and $\delta =  \exp(-\Omega(d))$, where the big-O notation hides the dependence on $\theta$. These guarantees have been partly extended, although only with a sub-optimal constant $\theta = O(1)$, to a broader class of random matrices that satisfy the Johnson-Lindenstrauss property, including dense subgaussian embeddings such as matrices with i.i.d.~random sign entries scaled by $1/\sqrt m$.

Dense Gaussian and subgaussian embeddings are too expensive for many applications, due to the high cost of dense matrix multiplication.  
One of the ways of addressing this, as proposed by Clarkson and Woodruff \cite{clarkson2013low}, is to use very sparse random sign matrices, where the sparsity is distributed uniformly so that there are $s\ll m$ non-zero entries per column of $S$ (we refer to $s$ as the column-sparsity of $S$). 
Remarkably, choosing column-sparsity $s=1$ (which is the minimum necessary sparsity for any OSE \cite{meng2013low}) is already sufficient to obtain a constant distortion OSE, but only if we increase the embedding dimension to $m=O(d^2)$. On the other hand, Nelson and Nguyen \cite{nelson2013osnap}, along with follow up works \cite{bourgain2015toward}, showed that if we allow $s=\polylog(d)$, then we can get an OSE with $m= d\polylog(d)$. This was later improved by Cohen \cite{cohen2016nearly} to $m=O(d\log(d))$ with column-sparsity $s= O(\log(d))$. Nevertheless, the embedding dimension of sparse OSEs remains sub-optimal, not just by a constant, but by an $O(\log(d))$ factor, due to a fundamental limitation of the matrix Chernoff analysis employed by \cite{cohen2016nearly}. Thus, we arrive at the central question of this work:

\bigskip

  \emph{What is the optimal embedding dimension for sparse oblivious subspace embeddings?}

\bigskip
  
This question is essentially the main open question posed by Nelson and Nguyen \cite[Conjecture 14]{nelson2013osnap}: They conjectured that a sparse random sign matrix with $s=O(\log(d))$ non-zeros per column achieves a constant distortion OSE with embedding dimension $m=O(d)$, i.e., within a constant factor of the optimum. We go one step further and ask whether a sparse OSE can recover the \emph{optimal} embedding dimension, i.e., $m=(1+\theta)d$ for any $\theta>0$, achieved by Gaussian embeddings. Note that this version of the question is open for \emph{any} sparsity $s$ and any distortion $\epsilon$, even including dense random sign matrices (i.e., $s=m$).

\subsection{Main results}

In our main result, we show that embedding matrices with column-sparsity $s$ polylogarithmic in $d$ can recover the optimal Gaussian embedding dimension $m=(1+\theta)d$, while achieving constant distortion $\epsilon$. The below result applies to several standard sparse embedding constructions, including a construction considered by Nelson and Nguyen (among others), where, we split each column of $S$ into $s$ sub-columns and sample a single non-zero entry in each sub-column, assigning a random $\pm1/\sqrt s$ to each of those entries.

\begin{theorem}[Sparse OSEs; informal Theorem \ref{osngeneral}]\label{t:direct}
Given any constant $\theta>0$, an $m\times n$ sparse embedding matrix $S$ with $n\geq m\geq (1+\theta)d$ and $s= O(\log^4(d))$ non-zeros per column is an oblivious subspace embedding with distortion $\epsilon = O(1)$. Moreover, given any $\epsilon,\delta$, it suffices to use $s=O(\log^4(d/\delta)/\epsilon^6)$ to get an $(\epsilon,\delta)$-OSE with  $m = O((d+\log1/\delta)/\epsilon^2)$.
\end{theorem}

The embedding dimension of $m=O(d/\epsilon^2)$ exactly matches a known lower bound of $m=\Omega(d/\epsilon^2)$, given by Nelson and Nguyen \cite{nelson2014lower}. To our knowledge, this result is the first to achieve the optimal embedding dimension for a sparse OSE with any sparsity $s=o(m)$, or indeed, for any OSE that can be applied faster than dense $d\times d$ matrix multiplication, including recent efforts \cite{chepurko2022near,cherapanamjeri2023optimal} (see Theorem~\ref{t:fast-ose} for our fast OSE algorithm).

A known lower bound on the level of sparsity achievable by any oblivious subspace embedding is a single non-zero entry per column ($s=1$) of the embedding matrix $S$ \cite{meng2013low}. However, this limit can be circumvented by \emph{non-oblivious} sparse embeddings, i.e., when we have additional information about the orthonormal matrix $U\in\R^{n\times d}$ that represents the subspace. Of particular significance is the distribution of the squared row norms of $U$ (also known as leverage scores), which encode the relative importance of the rows of $U$ in constructing a good embedding. Knowing accurate approximations of the leverage scores lies at the core of many subspace embedding techniques, including approximate leverage score sampling \cite{drineas2006sampling,drineas2012fast} and the Subsampled Randomized Hadamard Transform  \cite{ailon2009fast,tropp2011improved}. These approaches rely on the fact that simply sub-sampling $m=O(d\log d)$ rows of $U$ proportionally to their (approximate) leverage scores is a constant distortion subspace embedding. This corresponds to the $m\times n$ embedding matrix $S$ having one non-zero entry per \emph{row} (a.k.a.~a sub-sampling matrix), which is much sparser than any OSE since $n\gg m$. However, the sub-sampling embeddings are bound to the sub-optimal $O(d\log d)$ embedding dimension $m$, and it is not known when we can achieve the optimal $m=(1+\theta)d$ or even $m=O(d)$. We address this in the following result, showing that a non-oblivious sparse embedding knowing leverage score estimates requires only polylogarithmic in $d$ row-sparsity (non-zeros per row). To do this, we construct embedding matrices with a non-uniform sparsity pattern that favors high-leverage rows, inspired by recently proposed Leverage Score Sparsified embeddings \cite{less-embeddings,newton-less,gaussianization}.

\begin{theorem}[Sparser Non-oblivious SE; informal Theorem \ref{nonose}]\label{t:less}
    Consider $\alpha\geq 1$ and any matrix $U\in\R^{n\times d}$ such that $U^\top U=I$. Given $\alpha$-approximations of all squared row norms of $U$ (see Definition \ref{apprls}), we can construct a $(1+\theta)d\times n$ subspace embedding for $U$  having $O(\alpha\log^4 d)$ non-zeros per row and $\epsilon=O(1)$. Moreover, for any $\epsilon,\delta$, we can construct an $(\epsilon,\delta)$-SE for $U$ with $O(\alpha\log^4(d)/\epsilon^4)$ non-zeros per row and dimension $m=O((d+\log1/\delta)/\epsilon^2)$.    
\end{theorem}
Even though the above result focuses on non-oblivious embeddings, one of its key implications is a new guarantee for a classical family of oblivious embeddings known as Fast Johnson-Lindenstrauss Transforms (FJLT), introduced by Ailon and Chazelle \cite{ailon2009fast}. An FJLT is defined as $S=\Phi H D$, where $\Phi$ is an $m\times n$ uniformly sparsified embedding matrix, $H$ is an $n\times n$ orthogonal matrix with fast matrix-vector products (e.g., a Fourier or Walsh-Hadamard matrix), and $D$ is a diagonal matrix with random $\pm1$ entries. This embedding is effectively a two-step procedure: first, we use $HD$ to randomly rotate the subspace defined by $U$, obtaining $\tilde U=HDU$ which has nearly-uniform leverage scores \cite{tropp2011improved}; then we apply a uniformly sparsified embedding $\Phi$ to $\tilde U$, knowing that the uniform distribution is a good approximation for the leverage scores of $\tilde U$. Theorem~\ref{t:less} implies that an FJLT with $O(\log^4(d)/\epsilon^4)$ non-zeros per row is an $(\epsilon,\delta,d)$-OSE with the optimal embedding dimension $m=O((d+\log1/\delta)/\epsilon^2)$ (see Theorem \ref{usrht} for details). To our knowledge, this is the first optimal dimension OSE result for FJLT matrices.

Yet, the application to FJLTs does not leverage the full potential of Theorem~\ref{t:less}, which is particularly useful for efficiently constructing optimal subspace embeddings when only coarse leverage score estimates are available, which has arisen in recent works on fast subspace embeddings \cite{chepurko2022near,cherapanamjeri2023optimal}. In the following section, we use it to construct a new fast low-distortion SE (i.e., $\epsilon\ll 1$) with optimal embedding dimension, addressing a question posed by Cherapanamjeri, Silwal, Woodruff and Zhou \cite{cherapanamjeri2023optimal} (see Theorem \ref{t:fast-epsilon} for details).

\subsection{Fast subspace embeddings}
\label{s:fast}
Next, we illustrate how our main results can be used to construct fast subspace embeddings with optimal embedding dimension. In most applications, OSEs are used to perform dimensionality reduction on an $n\times d$ matrix $A$ by constructing the smaller $m\times d$ matrix $SA$. The subspace embedding condition ensures that $\|SAx\|\approx \|Ax\|$ for all $x\in\R^d$ up to a multiplicative factor $1+\epsilon$, which has numerous applications, including fast linear regression (see Section \ref{s:applications}). The key computational bottleneck here is the cost of computing $SA$. We aim for input sparsity time, i.e., $O(\nnz(A))$, where $\nnz(A)$ is the number of non-zeros in $A$, possibly with an additional small polynomial dependence on $d$. Our results for computing a fast subspace embedding $SA$ with optimal embedding dimension are summarized in Table~\ref{tab:fast}, alongside recent prior works.

\begin{table}
  \begin{tabular}{r||c|lc|l}
Ref. &Oblivious& \multicolumn{2}{l|}{Dimension $m$} & Runtime\\
  \hline\hline
\cite{chepurko2022near}&\cmark&$O(d\cdot\mathrm{pll}(d))$&--&
$O(\nnz(A)+d^{2+\gamma})$ 
\\
\cite{chepurko2022near}&--&$O(d\log(d)/\epsilon^2)$&--&
$O(\nnz(A)+d^\omega\,\mathrm{pll}(d))$ 
\\
\cite{cherapanamjeri2023optimal}& -- & $O(d)$&\cmark& $O(\nnz(A)+d^{2+\gamma})$ 
                                        \\
\cite{cherapanamjeri2023optimal}& -- & $O(d\log(d)/\epsilon^2)$&--& $O(\nnz(A)+d^{\omega})$\\                                 
   \hline
Thm. \ref{t:fast-ose}
          &\cmark&$O(d)$&\cmark&$O(\nnz(A)+d^{2+\gamma})$
                                        \\
  Thm. \ref{t:fast-epsilon}
          &--&$O(d/\epsilon^2)$&\cmark& $O(\nnz(A) + d^\omega)$
\end{tabular}
\vspace{3mm}
\caption{Comparison of our results to recent prior works
  towards obtaining fast subspace embeddings with optimal
  embedding dimension. For clarity of presentation, we assume that
  the distortion satisfies $\epsilon=\Omega(d^{-c})$ for a small
  constant $c>0$,
  and we use $\mathrm{pll}(d)$ to denote $\poly(\log\log d)$. We use a checkmark {\cmark} to indicate which embeddings are oblivious, and which of them achieve optimal dependence of dimension $m$ relative to the distortion $\epsilon$.}\label{tab:fast}
\end{table}

In the following result, we build on Theorem \ref{t:direct} to provide an input sparsity time algorithm for constructing a fast oblivious subspace embedding with constant distortion $\epsilon=O(1)$, and optimal embedding dimension $m=O(d)$. This is the first optimal OSE construction that is faster than current matrix multiplication time $O(d^\omega)$. 
\begin{theorem}[Fast oblivious subspace embedding]\label{t:fast-ose}
    Given $d\leq n$ and any $\gamma>0$, there is a distribution over
    $m\times n$ matrices $S$ where $m=O(d)$, such that for any
    $A\in\R^{n\times d}$, with probability at least $0.9$: 
    \begin{align*}
        \frac12\|Ax\|\leq\|SAx\|\leq 2\|Ax\|\qquad\forall x\in\R^d,
    \end{align*}
    and $SA$ can be computed in $O(\gamma^{-1}\nnz(A) +
    d^{2+\gamma}\polylog(d))$ time. Moreover, for $\gamma=\Omega(1)$, we can generate such a random matrix $S$ using only $\polylog(nd)$ many uniform random bits.
  \end{theorem}
  \begin{remark}\label{r:random-bits}
    The problem of constructing an OSE using $\polylog(nd)$ many random bits was also brought up by Nelson and Nguyen, who obtained this with $m=d\,\polylog(d)$. To achieve it with $m=O(d)$, we introduce a new sparse construction, likely of independent interest, where the non-zeros are distributed along the diagonals instead of the columns of $S$.
  \end{remark}
  The runtime of our method matches the best known non-oblivious SE with $m=O(d)$, recently obtained by \cite{cherapanamjeri2023optimal}, while at the same time being much simpler to implement: their construction requires solving a semidefinite program to achieve the optimal dimension, while we simply combine several sparse matrix multiplication steps. Moreover, thanks to its obliviousness, our embedding can be easily adapted to streaming settings. For example, consider numerical linear algebra in the turnstile model \cite{clarkson2009numerical}, where we wish to maintain a sketch of $A$ while receiving a sequence of updates $A_{i,j}\leftarrow A_{i,j}+\delta$. Using the construction from Theorem~\ref{t:fast-ose}, we can maintain a constant-distortion subspace embedding of $A$ in the turnstile model with optimal space of $O(d^2\log(nd))$ bits, while reducing the update time exponentially, from $O(d)$ (for a dense OSE matrix) to $\polylog(d)$ time.

  In the next result, we build on Theorem \ref{t:less} to provide the first subspace embedding with low distortion $\epsilon=o(1)$ and optimal embedding dimension $m=O(d/\epsilon^2)$ that can be applied in current matrix multiplication time. This addresses the question posed by Cherapanamjeri, Silwal, Woodruff and Zhou \cite{cherapanamjeri2023optimal}, who gave a current matrix multiplication time algorithm for a low-distortion subspace embedding, but with a sub-optimal embedding dimension $m=O(d\log(d)/\epsilon^2)$. We are able to improve upon this embedding dimension by replacing leverage score sampling (used by \cite{cherapanamjeri2023optimal}) with our leverage score sparsified embedding construction, developed as part of the proof of Theorem \ref{t:less}.
\begin{theorem}[Fast Low-distortion Subspace Embedding]\label{t:fast-epsilon}
    Given $A\in\R^{n\times d}$ and $\epsilon>0$, we can compute an $m\times d$ matrix
    $SA$ such that $m=O(d/\epsilon^2)$, and with probability at least $0.9$:
  \begin{align*}
    (1+\epsilon)^{-1}\|Ax\|\leq \|SAx\| \leq
    (1+\epsilon)\|Ax\|\qquad\forall x\in\R^d,
  \end{align*}
  in time $O(\gamma^{-1}\nnz(A) + d^\omega + \poly(1/\epsilon)d^{2+\gamma}\polylog(d))$ for
  any $0<\gamma<1$.    
\end{theorem}

\subsection{Applications to linear regression}
\label{s:applications}
 
Our fast subspace embeddings can be used to accelerate numerous approximation algorithms in randomized numerical linear algebra, including for linear regression, low-rank approximation, rank computation and more. Here, we illustrate this with the application to linear regression tasks. In the following result, we use Theorem \ref{t:fast-ose} to provide the first single pass algorithm for a relative error least squares approximation with optimal both time and space complexity. 
\begin{theorem}[Fast Least Squares]\label{t:least-squares}
  Given an $n\times d$ matrix $A$ and an $n\times 1$ vector $b$, specified with $O(\log nd)$-bit numbers, consider the task of finding $\tilde
  x$ such that:
  \begin{align*}
    \|A\tilde x-b\|_2 \leq (1+\epsilon) \min_x\|Ax-b\|_2.
  \end{align*}
  The following statements are true for this task:
  \begin{enumerate}
  \item For $\epsilon=\Theta(1)$, we can find $\tilde x$ with a
    single pass over $A$ and $b$ in $O(\nnz(A)+d^\omega)$ time, using
    $O(d^2\log(nd))$ bits of space.
  \item For arbitrary $\epsilon>0$, we can compute $\tilde x$
    in $O(\gamma^{-1}\nnz(A) + d^\omega +
d^{2+\gamma}/\epsilon)$ time, using
$O(d^2\log(nd))$ bits of space, for any $0<\gamma<1$.
\end{enumerate}
\end{theorem}
For part (1) of the claim, we note that the obtained space complexity matches the lower bound $\Omega(d^2\log (nd))$  of Clarkson and Woodruff \cite{clarkson2009numerical}. Moreover, it is clear that solving a least squares problem with any worst-case relative error guarantee requires at least reading the entire matrix $A$ and solving a $d\times d$ linear system, which implies that the $O(\nnz(A) + d^\omega)$ time is also optimal. For part (2) of the claim, we note that a similar time complexity for a $1+\epsilon$ (non-single-pass) least squares approximation was shown by \cite{cherapanamjeri2023optimal}, except they had an additional $O(\epsilon^{-1}d^2\polylog(d)\log(1/\epsilon))$. We avoid that extra term, thereby obtaining the correct $O(1/\epsilon)$ dependence on the relative error, by employing a carefully tuned preconditioned mini-batch stochastic gradient descent with approximate leverage score sampling. This approach is of independent interest, as it is very different from that of \cite{cherapanamjeri2023optimal}, who computed a sketch-and-solve estimate by running preconditioned gradient descent on the sketch.

Finally, we point out that our fast low-distortion subspace embeddings (Theorem \ref{t:fast-epsilon}) can be used to construct reductions for a wider class of constrained/regularized least squares problems, which includes Lasso regression among others \cite{bourgain2015toward}. The following result provides the first $O(d/\epsilon^2)\times d$ such reduction for $\epsilon=o(1)$ in current matrix multiplication time.
\begin{theorem}[Fast reduction for constrained/regularized least squares]\label{t:reduction}
   Given $A\in\R^{n\times d}$, $b\in\R^n$ and $\epsilon>0$, consider an $n\times d$ linear regression task $T(A,b,\epsilon)$ of finding $\tilde x$ such that:
    \begin{align*}
     f(\tilde x)\leq (1+\epsilon)\min_{x\in\mathcal C}
      f(x),\quad\text{where}\quad
      f(x) = \|Ax-b\|_2^2+ g(x),
    \end{align*}
    for some $g:\R^d\rightarrow\R_{\geq 0}$ and a set
    $\mathcal C\subseteq\R^d$. We can reduce this task to solving an
    $O(d/\epsilon^2)\times d$ instance $T(\tilde A,\tilde
    b,0.1\epsilon)$ in $O(\gamma^{-1}\nnz(A) + d^\omega +
    \poly(1/\epsilon)d^{2+\gamma}\polylog(d))$  time.
  \end{theorem}

\subsection{Overview of techniques}

One of the key ingredients in our analysis involves establishing the universality of a class of random matrices, building on the techniques of Brailovskaya and Van Handel \cite{brailovskaya2022universality}, by characterizing when the spectrum of a sum of independent random matrices is close to that of a Gaussian random matrix whose entries have the same mean and covariance. We adapt these techniques to a class of nearly-square random matrices that arise from applying an $m\times n$ sparse random matrix $S$ to an $n\times d$ isometric embedding matrix $U$, showing high probability bounds for the Hausdorff distance between the spectrum of $SU$ and the spectrum of a corresponding Gaussian random matrix. 

A key limitation of the results of \cite{brailovskaya2022universality} is that they require full independence between the random matrices in a sum (which correspond to sub-matrices of the matrix $S$), unlike, for instance, the analysis of Nelson and Nguyen \cite{nelson2013osnap} which uses a moment method that only requires $O(\log(d))$-wise independence. We address this with the \emph{independent diagonals} construction: we propose a distribution over $m\times n$ sparse random matrices $S$ where the non-zeros are densely packed into a small number of diagonals (see Figure \ref{fig:yosedia}) so that, while the diagonals are fully independent, the entries within a single diagonal only need to be 2-wise independent. As a consequence, the resulting construction requires only $n/m\cdot\polylog(n)$ uniform random bits to generate, and we further improve that to $\polylog(n)$ by combining it with the Nelson-Nguyen~embedding.
                      
Standard sparse embedding matrices are not very effective at producing low-distortion subspace embeddings, i.e., with $\epsilon=o(1)$, because their density (non-zeros per column) has to grow with $1/\epsilon$, so that their complexity is no longer input sparsity time. Prior work has dealt with this problem by using a constant distortion subspace embedding as a preconditioner for computing the leverage score estimates $l_1,...,l_n$ of the input matrix $A$ \cite{chepurko2022near}, and then constructing a subspace embedding in a non-oblivious way out of a sub-sample of $m=O(d\log d/\epsilon^2)$ rows of $A$. This leverage score sampling scheme is effectively equivalent to using an extremely sparse embedding matrix $S$ which has a single non-zero entry $S_{i,I_i}\sim \pm 1/\sqrt{l_i}$ in each row, with its index $I_i$ sampled according to the leverage score distribution $(l_1/Z,...,l_n/Z)$, where $Z=\sum_il_i$. Unfortunately due to the well-known coupon collector problem, such a sparse embedding matrix cannot achieve the optimal embedding dimension $m=O(d/\epsilon^2)$. We circumvent this issue by making the embedding matrix $S$ slightly denser, with $\alpha\poly(1/\epsilon)\polylog(d)$ non-zeros per row, where $\alpha$ is the approximation factor in the leverage score distribution (i.e., \emph{leverage score sparsification}, see Figure \ref{fig:nseindent}). Unlike the oblivious sparse embedding, here it is the row-density (instead of column-density) that grows with $1/\epsilon$, which means that the overall algorithm can still run in input sparsity time. We note that our algorithms use $\alpha=O(d^\gamma)$ approximation factor for the leverage scores, where $0<\gamma<1$ is a parameter that can be chosen arbitrarily. This parameter reflects a trade-off in the runtime complexity, between the $O(\gamma^{-1}(\nnz(A)+d^2))$ cost of estimating the leverage scores, and the density of the leverage score sparsified embedding.

To construct a least squares approximation with $\epsilon=o(1)$ (Theorem \ref{t:least-squares} part 2), we use our constant distortion subspace embedding to compute a preconditioner for matrix $A$. That preconditioner is then used first to approximate the leverage scores, as well as to compute a constant factor least squares approximation, and then to improve the convergence rate of a gradient descent-type algorithm. However, unlike prior works \cite{chepurko2022near,cherapanamjeri2023optimal,yang2017weighted}, which either use a full gradient or a stochastic gradient based on a single row-sample, we observe that the computationally optimal strategy is to use a stochastic gradient based on a mini-batch of $O(\alpha d)$ rows, where $\alpha$ is the leverage score approximation factor. With the right sequence of decaying step sizes, this strategy leads to the optimal balance between the cost of computing the gradient estimate and the cost of preconditioning it, while retaining fast per-iteration convergence rate, leading to the $O(\alpha d^2/\epsilon)$ overall complexity of stochastic gradient descent. 

\subsection{Related work}
Our results follow a long line of work on matrix sketching techniques, which have emerged as part of the broader area of randomized linear algebra; see  \cite{halko2011finding,woodruff2014sketching,drineas2016randnla,martinsson2020randomized,derezinski2021determinantal,randlapack_book} for comprehensive overviews of the topic. These methods have proven pivotal in speeding up fundamental linear algebra tasks such as least squares regression \cite{sarlos2006improved,rokhlin2008fast,clarkson2013low}, $l_p$ regression \cite{meng2013low,cohen2015lp,chen2021query,wang2022tight}, low-rank approximation \cite{cohen2017input,li2020input}, linear programming \cite{cohen2021solving}, and more \cite{song2019relative,jiang2020faster}. Many of these results have also been studied in the streaming and turnstile models \cite{clarkson2009numerical}.

Subspace embeddings are one of the key algorithmic tools in many of the above randomized linear algebra algorithms. Using sparse random matrices for this purpose was first proposed by Clarkson and Woodruff \cite{clarkson2013low}, via the CountSketch which has a single non-zero entry per column, and then further developed by several other works \cite{nelson2013osnap,meng2013low,cohen2016nearly} to allow multiple non-zeros per column as well as refining the embedding guarantees. Non-uniformly sparsified embedding constructions have been studied recently, including Leverage Score Sparsified embeddings \cite{less-embeddings,newton-less,gaussianization,derezinski2022sharp}, although these works use much denser matrices than we propose in this work, as well as relying on somewhat different constructions. There have also been recent efforts on achieving the optimal embedding dimension for subspace embeddings, including \cite{cartis2021hashing}, who also rely on sparse embeddings, but require additional assumptions on the dimensions of the input matrix as well as its leverage score distribution; and \cite{chepurko2022near,cherapanamjeri2023optimal}, who do not rely on sparse embedding matrices, and therefore do not address the conjecture of Nelson and Nguyen (see Table \ref{tab:fast} for a comparison).

%% file: preliminaries.tex
\section{Preliminaries} \label{sec:prelim}

\paragraph{\textit{Notation}} The following notation and terminology will be used in the paper. The notation $[n]$ is used for the set $\{1,2,...,n\}$. All matrices considered in this work are real valued and the space of $m \times n$ matrices with real valued entries is denoted by $M_{m \times n}(\mathbb{R})$. The operator norm of a matrix $X$ as $\norm{X}$ and its condition number by $\kappa(X)$. For clarity, the operator norm is also denoted by $\norm{X}_{op}$ in some places where other norms appear. We shall denote the spectrum of a matrix $X$, which is the set of all eigenvalues of $X$, by $\spec(X)$. The standard probability measure is denoted by $\mathbb{P}$, and the symbol $\mathbb{E}$ means taking the expectation with respect to the probability measure. The standard $L_p$ norm of a random variable $\xi$ is denoted by $\norm{\xi}_p$, for $1 \le p \le \infty$. Throughout the paper, the symbols $c_1, c_2, ...$, and $Const, Const', ...$ denote absolute constants. 

\bigskip

\paragraph{\textit{Oblivious Subspace Embeddings}} 
We define an oblivious subspace embedding, i.e., an $(\epsilon,\delta,d)$-OSE, following Definition \ref{def:OSE}, to be a random $m\times n$ matrix $S$ such that for any $n\times d$ matrix $U$ with orthonormal columns (i.e., $U^\top U=I_d$), 
\begin{align}\label{eq:ose}
\Pb\big(\,\forall_{x\in \R^d}\ 
    (1+\epsilon)^{-1}\|x\|\leq\|SUx\|\leq(1+\epsilon)\|x\|\,\big)\geq 1-\delta.        
\end{align}

For computational efficiency, we usually consider sparse OSEs. A standard construction for a sparse OSE involves i.i.d. rademacher entries in each position, sparsified by multiplication with independent Bernoulli random variables. More precisely, $S$ has i.i.d. entries $s_{i,j}=\delta_{i,j} \xi_{i,j}$ where $\delta_{i,j}$ are independent Bernoulli random variables taking value 1 with probability $p_{m,n,d} \in (0,1]$ and $\xi_{i,j}$ are i.i.d. random variables with $\Pb(\xi_{i,j}=1)=\Pb(\xi_{i,j}=-1)=1/2$. Note that this results in $S$ having $s = pm$ many non zero entries per column and $pn$ many non zero entries per row on average. We shall call this the oblivious subspace embedding with independent entries distribution.  

\begin{definition}[OSE-IID-ENT]\label{oseie}
A $m \times n$ random matrix $S$ is called an oblivious subspace embedding with independent entries (OSE-IID-ENT) if $S$ has i.i.d. entries $s_{i,j}=\delta_{i,j} \xi_{i,j}$ where $\delta_{i,j}$ are independent Bernoulli random variables taking value 1 with probability $p_{m,n,d} \in (0,1]$ and $\xi_{i,j}$ are i.i.d. random variables with $\Pb(\xi_{i,j}=1)=\Pb(\xi_{i,j}=-1)=1/2$.
\end{definition}

Another example comes from a class of sparse sketching matrices proposed by Nelson and Nguyen \cite{nelson2013osnap}, called OSNAPs. They define a sketching matrix $S$ as an \textit{oblivious sparse norm-approximating projection} (OSNAP) if it satisfies the following properties - 
\begin{enumerate}
    \item $s_{ij} = \delta_{ij}\sigma_{ij}/\sqrt{s}$ where $\sigma$ are i.i.d. $\pm 1$ random variables, and $\delta_{ij}$ is an indicator random variable for the event $S_{ij} \neq 0$.
    \item $\forall j \in [n], \sum_{i=1}^{m} \delta_{i,j} = s$ with probability $1$, i.e. every column has exactly $s$ non-zero entries.
    \item For any $T \subset [m]\times [n], \E \Pi_{(i,j) \in T} \delta_{ij} \leq (s/m)^{|T|}$.
    \item The columns of $S$ are i.i.d.
\end{enumerate}

One example of an OSNAP can be constructed as follows when $s$ divides $m$. In this case, we divide each column of $S$ into $s$ many blocks, with each block having $\frac{m}{s}$ many rows. For each block, we randomly and uniformly select one nonzero entry and set its value to be $\pm 1$ with probability $1/2$. Note that the blocks in each column are i.i.d., and the columns of $S$ are i.i.d. We then see that $S/\sqrt{s}$ satisfies the properties of an OSNAP. For convenience, in this work we will refer to this as the OSNAP distribution, and we will define it using the parameter $p=s/m$ instead of $s$. To define such a distribution formally, we first define the one hot distribution.

\begin{definition}[One Hot Distribution] \label{def:ohd}
    Let $M$ be an $a \times b$ random matrix. Let $\gamma$ be a random variable taking values in $[a] \times [b]$ with $\Pb(\gamma = (i,j)) = (1/ab)$. Let $\xi$ be a Rademacher random variable ($\Pb(\xi=-1)=\Pb(\xi=1)=\frac{1}{2}$). $M$ is said to have the one hot distribution if $M=\xi(\sum \limits_{(i,j) \in [a] \times [b]} \mathbbm{1}_{\{(i,j)\}}(\gamma)E_{i,j})$ where $E_{i,j}$ is an $a \times b$ matrix with $1$ in $(i,j)^{th}$ entry and $0$ everywhere else. 
\end{definition}

\begin{definition}[OSNAP-IND-COL]\label{osnap} 
An $m \times n$ random matrix $S$ is called an oblivious sparse norm-approximating projection with independent subcolumns distribution (OSNAP-IND-COL)
with parameter $p$ such that $s=pm$ divides $m$, if each submatrix $S_{[(m/s)(i-1)+1:(m/s)i]\times\{ j\}}$ of $S$ for $i\in [s], j \in [n]$ has the one hot distribution, and all these submatrices are jointly independent. 
\end{definition}
 Below we collect the existing subspace embedding results for OSNAP matrices, which are relevant to this work.
  \begin{lemma}[Existing Sparse Embedding Guarantees]\label{l:existing}
    The following are some of the known guarantees for OSNAP
    embedding matrices:    
    \begin{itemize}
    \item \cite{clarkson2013low} showed that there is an OSNAP matrix $S$ with $m=O(\epsilon^{-2}d^2)$ rows and $1$ non-zero per column (i.e., CountSketch) which is an OSE with distortion $\epsilon$. 
  \item \cite{cohen2016nearly} showed that there is an OSNAP matrix $S$ with
    $m=O(\epsilon^{-2}d^{1+\gamma}\log d)$ rows and $s=O(1/\gamma\epsilon)$
    non-zero entries per column which is an OSE with distortion $\epsilon$.
    Note that setting $\gamma = 1/\log(d)$, we get
      $m=O(\epsilon^{-2}d\log d)$ and $O(\log(d)/\epsilon)$ non-zeros
      per column.
  \item \cite{nelson2013osnap} showed that there is an OSNAP matrix
    using $O(\log(d)\log(nd))$ uniform random bits with
    $m=O(\epsilon^{-2}d^{1+\gamma}\log^8(d))$ and $O(1/\gamma^3\epsilon)$ non-zero entries per column.
  \end{itemize}
  \end{lemma}

\bigskip

\paragraph{\textit{Non-oblivious subspace embeddings}}
Following Definition \ref{def:OSE}, we say that an $m\times n$ random matrix $S$ is a (non-oblivious) subspace embedding for a given $n\times d$ matrix $U$ with orthonormal columns if it satisfies \eqref{eq:ose} for that matrix $U$.
In this case, to obtain subspace embedding guarantees with even sparser random matrices, we can use the information about the subspace in the form of its leverage scores. For $i=1,...,n$, the $i$th leverage score of a $d$-dimensional subspace of $\R^n$ is the squared norm of the $i$th row of its orthonormal basis matrix $U$, i.e., $\|e_i^\top U\|_2^2$ (this definition is in fact independent of the choice of basis).

We note that in most applications (e.g., Theorem \ref{t:fast-epsilon}), subspace embedding matrices are typically used to transform an arbitrary $n\times d$ matrix $A$ (not necessarily with orthonormal columns), constructing a smaller $m\times d$ matrix $SA$. In this case, we seek an embedding for the subspace of vectors $\{z: z=Ax\text{ for } x\in\R^d\}$. Here, the corresponding $U$ matrix has columns that form an orthonormal basis for the column-span of $A$. Thus, in practice we do not have access to matrix $U$ or its leverage scores. Instead, we may compute leverage score approximations \cite{drineas2012fast}.

\begin{definition}[Approximate Leverage Scores]\label{apprls}
    For $\beta_1 \ge 1, \beta_2 \ge 1$, a tuple $(l_1, \ldots, l_n)$ of numbers are $(\beta_1,\beta_2)$-approximate leverage scores for $U$ if, for $1\leq i\leq n$,
    \begin{align*}
        \frac{\norm{e_i^TU}^2}{\beta_1} \leq l_i \qquad\text{and}\qquad
    \sum_{i=1}^n l_i \leq \beta_2 (\sum_{i=1}^n \norm{e_i^TU}^2) = \beta_2 d.
    \end{align*}
    And in this case, we also say that they are $\alpha$-approximations of squared row norms of $U$ with $\alpha=\beta_1\beta_2$.
\end{definition}

\bigskip

\paragraph{\textit{Uniformizing leverage scores by preconditioning}} 
Another way of utilizing information about the leverage scores to get embedding guarantees with sparser matrices is to precondition the matrix $U$ using the randomized Hadamard transform to uniformize the row norms, resulting in $(d/n, d/n, \ldots d/n)$ becoming approximate leverage scores for the preconditioned matrix. To this end, we first define the Walsh-Hadamard matrix.

\begin{definition} \label{hada} The Walsh-Hadamard matrix $H_{2^k}$ of dimension $2^k \times 2^k$ for $k \in \mathbb{N} \cup \{ 0 \}$ is the matrix obtained using the recurrence relation
\begin{align*}
    H_0 &= [1],\qquad
    H_{2n} = \begin{bmatrix}
        H_n & H_n \\
        H_n & -H_n
    \end{bmatrix}.
\end{align*}
\end{definition}
In what follows, we drop the subscript of $H_{2^k}$ when the dimension is clear.

\begin{definition} \label{def:RHT}
The \textit{randomized Hadamard transform} (RHT) of an $n \times d$ matrix $U$ is the product $\frac{1}{\sqrt{n}}HDU$, where $D$ is a random $n \times n$ diagonal matrix whose entries are independent random signs, i.e., random variables uniformly distributed on $\{\pm 1\}$. Here, by padding $U$ with zero rows if necessary, we may assume that $n$ is a power of 2. 
\end{definition}

The key property of the randomized Hadamard transform that we use is that it uniformizes the row norms of $U$ with high probability. More precisely, we have, 

\begin{lemma}[Lemma 3.3, \cite{tropp2011improved}]\label{rhtrownorms}
    Let $U$ be an $n \times d$ matrix with orthonormal columns. Then, $\frac{1}{\sqrt{n}}HDU$ is
an $n \times d$ matrix with orthonormal columns, and, for $\delta > 0$
\begin{align*}
\Pb \left( \max \limits_{j=1,...,n} \norm{e_j^T(\frac{1}{\sqrt{n}}HDU)} \geq \sqrt{\frac{d}{n}}+\sqrt{\frac{8 \log( n/\delta)}{n}} \right) \leq \delta
\end{align*}
\end{lemma}

\bigskip

\paragraph{\textit{Universality}} In this paragraph, we describe the random matrix universality result of \cite{brailovskaya2022universality}, which is central to our analysis of sparse subspace embedding matrices. The object of study here is a random matrix model given by 
\begin{align}
     X:= Z_0 + \sum_{i=1}^{n} Z_i \label{eq:model}
\end{align}
where $Z_0$ is a symmetric deterministic $d \times d$ matrix and $Z_1, \ldots, Z_n$ are symmetric independent random matrices with $\E[Z_i]=0$.
We shall compare the spectrum of $X$ to the spectrum of a gaussian model $G$ that has the same mean and covariance structure as $X$. More precisely, denoting by $\cov(X)$ the $d^2 \times d^2$ covariance matrix of the entries of $X$, 
\[ \cov(X)_{i,j,k,l} := \E[(X - \E X)_{ij}(X-\E X)_{kl}]\]
$G$ is the $d \times d$ symmetric random matrix such that:
\begin{enumerate}
    \item $\{ G_{ij}: i,j \in [d] \}$ are jointly Gaussian
    \item $\E[G] = \E[X]$ and $\cov(G) = \cov(X)$.
\end{enumerate}
The above two properties uniquely define the distribution of $G$. We next define the notion of Hausdorff distance, which will be used in the universality result below.
\begin{definition}[Hausdorff distance] Let $A, B \subset \R^n$. Then the Hausdorff distance between $A$ and $B$ is given by,
\begin{align*}
    d_H(A,B) = \inf \{ \varepsilon \geq 0 ; A \subseteq B_{\varepsilon} \text{ and } B \subseteq A_{\varepsilon} \}
\end{align*}
 where $A_{\varepsilon}$ (resp. $B_{\varepsilon}$) denotes the $\varepsilon$-neighbourhood of $A$.
\end{definition}

\begin{lemma}[Theorem 2.4 \cite{brailovskaya2022universality}] \label{thm:brailovskayavanhandel}
Given the random matrix model \eqref{eq:model}, define the following:
\begin{align*}
	\sigma(X) & = \norm{\E[(X-\E X)^2]}^{\frac{1}{2}}_{op}\\
	\sigma_*(X) &= 
	\sup \limits _{\|v\|=\|w\|=1} \E[
	|\langle v,(X-\E X)w\rangle|^2]^{\frac{1}{2}}\\
\text{and}\quad	R(X) &= \norm{\max_{1\le i\le n}\norm{Z_i}_{op}}_\infty.
\end{align*}
There is a universal constant $C>0$ such that for any $t \geq 0$,
  \begin{align*}
\Pb\Big(d_H\big(\spec(X),\spec(G)\big)>C\epsilon(t)\Big)\leq de^{-t},\\
\text{where}\quad \epsilon(t) =
	\sigma_*(X)  t^{\frac{1}{2}} +
	R(X)^{\frac{1}{3}}\sigma(X)^{\frac{2}{3}} t^{\frac{2}{3}} +
	R(X) t.
\end{align*}
\end{lemma}

This result can be viewed as a sharper version of the Matrix Bernstein inequality \cite{tropp2011improved} for the concentration of sums of random matrices. To see this, note that for the random matrix model \eqref{eq:model}, Matrix Bernstein implies that:
\begin{align*}
\E \norm{X} \lesssim \sigma(X)\sqrt{\log d}+R(X)\log d,
\end{align*}
which can be recovered by Lemma \ref{thm:brailovskayavanhandel} (see Example 2.12 in \cite{brailovskaya2022universality}).
However, Lemma \ref{thm:brailovskayavanhandel} together with Theorem 1.2 in \cite{bandeira2023matrix} implies that:
\begin{align*}
\E(\norm{X}) \le C(\sigma(X)+v(X)^{1/2}\sigma(X)^{1/2}(\log d)^{3/4}+R(X)^{\frac{1}{3}}\sigma(X)^{\frac{2}{3}} (\log d)^{2/3}+R(X) \log d)
\end{align*}
where $v(X)=\norm{\cov(X)}$ is the norm of the covariance matrix of the $d^2$ scalar entries. This result can be sharper than the Matrix Bernstein inequality because when $v(X)$ and $R(X)$ are small enough, then we will have $\E \norm{X} \lesssim \sigma(X)$, which improves the Matrix Bernstein inequality by removing the $\sqrt{\log d}$ factor.

\bigskip

\paragraph{\textit{Spectrum of Gaussian Matrices}} To leverage the universality properties of the random matrix model, we shall rely on the following result about the singular values of Gaussian matrices, which in particular can be used to recover the optimal subspace embedding guarantee for Gaussian sketches.

\begin{lemma}[(2.3), \cite{rudelson2010non}]\label{lem:Gaussianspectrum}
    Let $G$ be an $m \times n$ matrix whose entries are independent standard normal variables. Then,
    \begin{align*}
        \Pb(\sqrt{m}-\sqrt{n}-t \leq s_{min}(G) \leq s_{\max}(G) \leq \sqrt{m}+\sqrt{n}+t) \geq 1 - 2e^{-t^2/2}
    \end{align*}
\end{lemma}

%% file: oblivious.tex
\section{Analysis of Oblivious Sparse Embeddings}\label{s:analysis-direct}

In this section, we state and prove our main OSE result, Theorem \ref{osngeneral} (given as Theorem~\ref{t:direct} in Section \ref{s:intro}).
Before we get to the proof, however, we propose a new model for sparse OSEs that is designed to exploit the strength of our proof in dealing with the number of independent random bits required. 

To illustrate the issue, consider that in the OSE-IID-ENT model (Definiton \ref{oseie}), we need $mn$ many random bits to determine the nonzero entries of the matrix $S$, even though the matrix will have far fewer non-zero entries. Naturally, there are many known strategies for improving this, including OSNAP-IND-COL (Definition \ref{osnap}) and even more elaborate hashing constructions based on polynomials over finite fields \cite{nelson2013osnap}, which allow reducing the random bit complexity to polylogarithmic in the dimensions. However, these constructions do not provide sufficient independence needed in the random matrix model \eqref{eq:model} to apply the universality result of Brailovskaya and Van Handel (Lemma~\ref{thm:brailovskayavanhandel}). We address this with the \emph{independent diagonals} distribution family, defined shortly.

In order to use universality for establishing a subspace embedding guarantee, we must analyze a symmetrized version of the matrix $SU$, for an $n\times d$ orthogonal matrix $U$, with $S$ being a sum of sparse independent random matrices, say $Y_i$'s. Naturally, to compare the spectra of $SU$ and an appropriate Gaussian model, the matrix $S$ cannot be too sparse. However, since the individual entries of each $Y_i$ need not be independent, we can reduce the number of independent summands in $S$ by making each individual summand $Y_i$ denser. At the same time, we need to control $\norm{Y_iU}$, just as we would when using the standard matrix Bernstein inequality.

Both these goals can be achieved by placing non-zero entries along a diagonal of $Y_i$. Placing $\pm 1$ entries along a diagonal of a matrix keeps its norm bounded by $1$ whereas other arrangements (say, along a row or column) do not. Moreover, these $\pm 1$ entries along the diagonal need not be independent, they can simply be uncorrelated. As a result, an instance of $Y_i$ can be generated with just $O(1)$ random bits.

\subsection{Independent Diagonals Construction}
With this motivation, we define the independent diagonals distribution formally (Figure \ref{fig:yosedia} illustrates this construction).

\begin{definition}[OSE-IND-DIAG]\label{oseinddiag}
An $m \times n$ random matrix $S$ is called an oblivious subspace embedding with independent diagonals (OSE-IND-DIAG) with parameter $p$ if it is constructed in the following way. Assume that $np$ is an integer. Let $W=(w_1, \ldots, w_m)$ be a random vector whose components are $\pm 1$ valued and uncorrelated, i.e. $\E[w_i] = 0, \E[w_i^2] = 1, \E[w_iw_j]=0$. We define $\gamma$ to be a random variable uniformly distributed in $ [n]$. Let $\gamma_1,...,\gamma_{np}$ be i.i.d. copies of $\gamma$. Let $W_1,...,W_{np}$ be i.i.d. copies of $W$. Let $F_{j}(x)$ be a function that transforms a $m$ dimensional vector $x$ to the $m\times n$ matrix putting $x$ on the $j\textsuperscript{th}$ diagonal, i.e. positions $(1,j)$ through $(m,j+m \mod n)$ with all other entries zero (See Fig \ref{fig:yosedia} for an illustration). Let $S=\sum \limits_{l \in [np]}F_{\gamma_l}(W_l)$. 
\end{definition}

\begin{figure}
    \centering
    \hspace{-3mm}\begin{minipage}[L]{0.5\linewidth}\includegraphics[width=\linewidth]{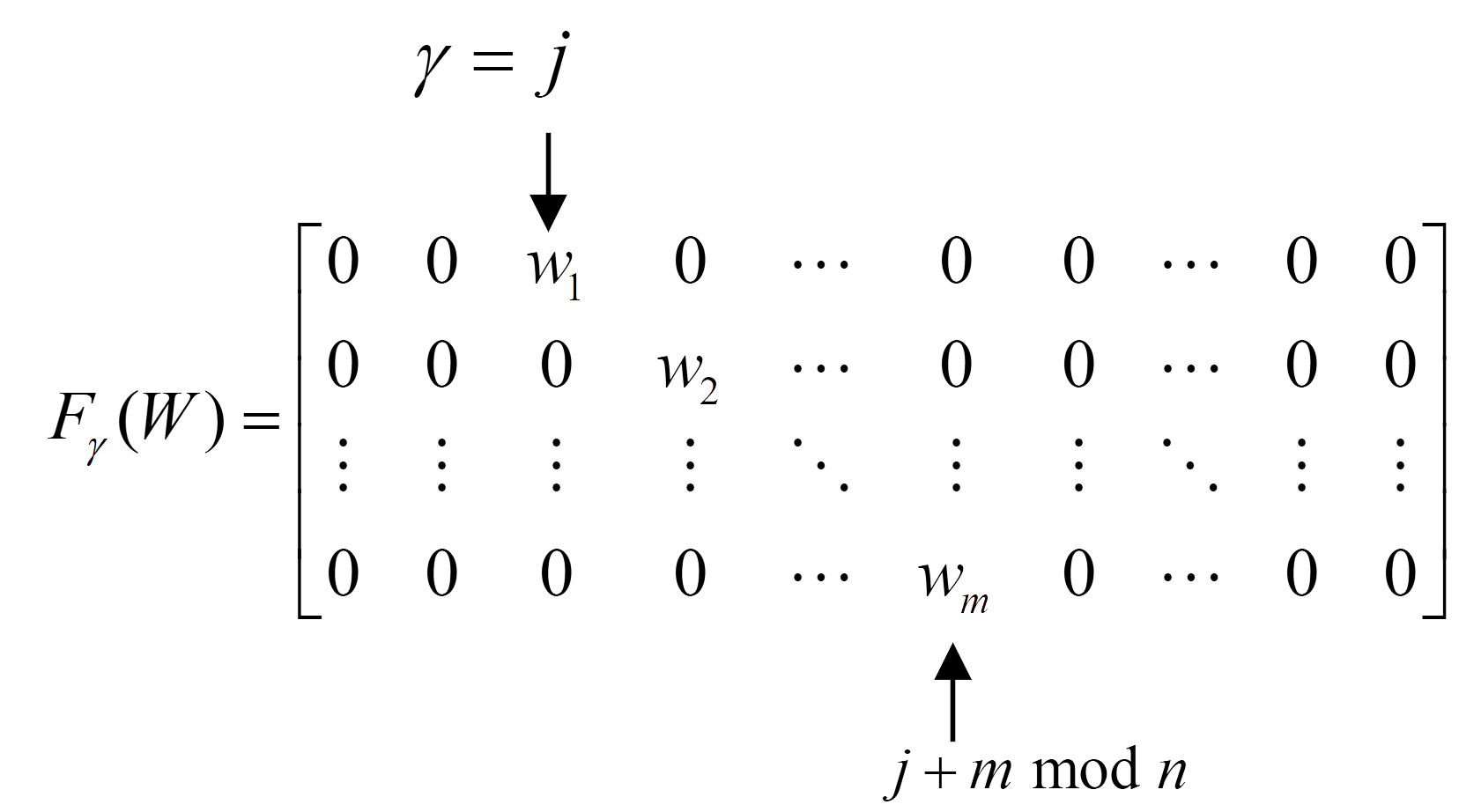}\end{minipage}\hspace{3mm}%
    \begin{minipage}[R]{0.5\linewidth}\includegraphics[width=\linewidth]{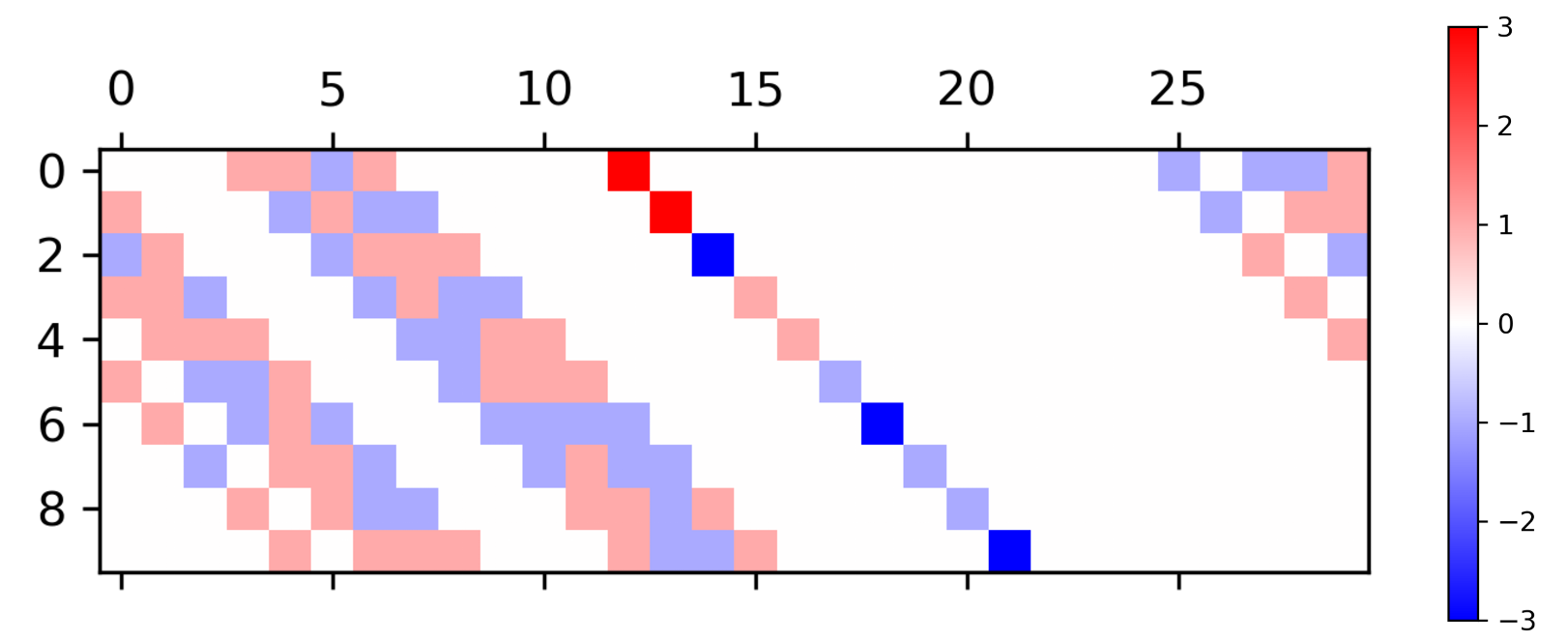}
    \end{minipage}

    \caption{Left: Structure of the random matrix $F_{\gamma_l}(W_l)$. Right: Illustration of an embedding matrix $S$ with independent diagonals with parameters $d=8$, $m=10$, and $n=30$, showing how the nonzero entries of $S$ occur along diagonals. The number of diagonals is controlled by the parameter $p$.}
    \label{fig:yosedia}
\end{figure}

Universality results show that the properties of a general random matrix are similar to the properties of a gaussian random matrix with the same covariance profile. Therefore, to analyze the OSE models, we need to first calculate the covariances between entries.

\begin{lemma}[Variance and Uncorrelatedness] \label{lem:entries} 
Let $p = p_{m,n} \in (0,1]$ and $S=\{s_{ij}\}_{i \in [m], j \in [n]}$ be a $m \times n$ random matrix distributed according to the OSE-IID-ENT, OSNAP-IND-COL, or OSE-IND-DIAG distributions. Then, $\E(s_{ij})=0$ and $\operatorname{Var}(s_{ij})=p$ for all $i \in [m], j \in [n]$, and $\cov(s_{i_1 j_1},s_{i_2 j_2})=0$ for any $\{i_1,i_2\} \subset [m], \{j_1,j_2\} \subset [n]$ and $ (i_1,j_1)\neq (i_2,j_2) $
\end{lemma}
\begin{proof}
For i.i.d. entries case, the result is straightforward.

For the independent diagonals case, since $S$ is a sum of i.i.d. copies of $Y := F_{\gamma}(W)$, we analyze $Y$ first. By definition,
\begin{align*}
    Y = \sum_{i=1}^{m} w_i \left( \sum_{j=1}^{n} \mathbbm{1}_{\{j-i+1\}}(\gamma)E_{i,j} \right)
\end{align*}
where $W=(w_1, \ldots, w_m)$ be a random vector whose components are $\pm 1$ valued and uncorrelated, i.e. $\E[w_i] = 0, \E[w_i^2] = 1, \E[w_iw_j]=0$ and $j-i+1$ is considered modulo $n$ in $[n]$. Thus,
\begin{align*}
    y_{ij} = w_i \mathbbm{1}_{\{j-i+1\}}(\gamma)
\end{align*}
First, note that we have $\E(y_{ij})=0$, and,
\begin{align*}
\E(y_{ij}^2)=(\E(\mathbbm{1}_{\{j-i+1\}}(\gamma)^2)\E(w_i^2)=\frac{1}{n}
\end{align*}
Moreover,
\begin{align*}
    \cov(y_{i_1j_1}, y_{i_2 j_2}) &= \E(y_{i_1 j_1}y_{i_2 j_2}) \\
    &= \E(w_i w_k \mathbbm{1}_{\{j_1-i_1+1\}}(\gamma)\mathbbm{1}_{\{j_2-i_2+1\}}(\gamma))
    \\ &= \E(w_{i_1}w_{i_2})\E(\mathbbm{1}_{\{j_1-i_1+1\}}(\gamma)\mathbbm{1}_{\{j_2-i_2+1\}}(\gamma))
\end{align*}
by independence between $\gamma$ and $W$.

When $j_1-i_1 \neq j_2-i_2$, we know that $\mathbbm{1}_{\{j_1-i_1+1\}}(\gamma)\mathbbm{1}_{\{j_2-i_2+1\}}(\gamma)$ is always $0$ because $\gamma$ can only take one value.

When $j_1-i_1 = j_2-i_2$, since $(i_1,j_1)\neq (i_2,j_2) $, we must have $i_1 \ne i_2$. In this case, by uncorrelatedness of the components of $W$,
\begin{align*}
    \cov(y_{i_1 j_1}, y_{i_2 j_2}) = 0 \cdot \E(\mathbbm{1}_{\{j_1-i_1+1\}}(\gamma)\mathbbm{1}_{\{j_2-i_2+1\}}(\gamma)) = 0 
\end{align*}
Now, since \begin{align*}S=\sum\limits_{l=1}^{np}Y_l \text{,}\end{align*}we have, by using the independence of $Y_l$, $\Var(s_{i_1 j_1})=p$, and, when $(i_1, j_1)\neq (i_2, j_2)$,
\begin{align*}
    \cov(s_{i_1 j_1}, s_{i_2 j_2}) &= \E(s_{i_1 j_1}s_{i_2 j_2}) = \sum_{l=1}^{np} \E \left( (Y_l)_{i_1 j_1}(Y_l)_{i_2 j_2} \right) = 0
\end{align*}

We now consider the independent subcolumns distribution OSNAP-IND-COL. Recall that in the OSNAP-IND-COL distribution, each subcolumn $S_{[(m/s)(i-1)+1:(m/s)i]\times\{ j\}}$ of $S$ for $i\in [s], j \in [n]$ has the one hot distribution, and all these subcolumns are jointly independent. Therefore, the same argument as above, we directly calculate the variances of the entries
\begin{align*}
    \E (s_{ij}^2) = \Pb(s_{ij} \neq 0) = \frac{1}{\frac{m}{pm}} = p
\end{align*}

For the covariances, we first observe that $\cov(s_{i_1 j_1}, s_{i_2 j_2})=0$ if $(i_1,j_1)$ and $(i_2,j_2)$ belong to two different subcolumns by independence. If $(i_1,j_1)$ and $(i_2,j_2)$ belong to the same subcolumn, we have
\begin{align*}
    \cov(s_{i_1 j_1}, s_{i_2 j_2})=\E(s_{i_1 j_2}s_{i_2 j_1})=\E(0)=0
\end{align*}
because each subcolumn has one hot distribution which means at most one of $s_{i_1 j_1}$ and $s_{i_2 j_2}$ can be nonzero.

\end{proof}

\begin{lemma}[Random Bit Complexity]
    For a parameter $0<p<1$, the random bit complexity to generate the OSE-IND-DIAG distribution  is $O(pn \log(n))$.
\end{lemma}
\begin{proof}
    We can generate the vector $W$ using $\log(m)$ random bits, using a construction involving linear functions on finite fields, to get $m$ 2-wise independent random variables (Prop. 3.24 and Corollary 3.34. in \cite{vadhan2012pseudorandomness}) and then transforming them to be $\pm 1$ valued. To generate the random variable $\gamma$ which is uniformly distributed in $[n]$, the number of random bits needed is proportional to the information entropy of $\gamma$, which is $O(\log(n))$ by \cite[Theorem 5.11.2]{cover2006elements}. Taken together, the total number of random bits needed to generate each $F_{\gamma_l}(W_l)$ is $O(\log(m)+\log(n))$.
    
    Finally, since $S$ is the sum of $np$ many independent copies of $F_{\gamma_l}(W_l)$, it takes $O((\log(n)+\log(m)) \cdot np )=O(pn \log(n))$ many random bits to generate $S$.    
\end{proof}

\subsection{Proof of Theorem \ref{t:direct}}

We now state the main theorem of this section, which is the detailed version of Theorem \ref{t:direct}.
\begin{theorem}[Analysis of OSE by Universality]\label{osngeneral}
Let $S$ be an $m \times n$ matrix distributed according to the OSE-IID-ENT, OSNAP-IND-COL, or OSE-IND-DIAG distributions with parameter $p$. Let $U$ be an arbitrary $n \times d$ deterministic matrix such that $U^TU=I$. Then, there exist constants $c_{\ref{osngeneral}.1}$ and $c_{\ref{osngeneral}.2}$ such that for any $\varepsilon, \delta > 0$, we have 
\begin{align*}
\Pb \left( 1 - \varepsilon  \leq s_{\min}((1/\sqrt{pm})SU)   \leq s_{\max}((1/\sqrt{pm})SU) \leq 1 + \varepsilon \right) \geq 1-\delta
\end{align*}
when $m \geq c_{\ref{osngeneral}.1}  \max (d, \log(4/\delta))/\varepsilon^2$ and $pm>c_{\ref{osngeneral}.2}{(\log (d/\delta ))^4}/{\varepsilon ^6}$.

Alternatively, given a fixed $\theta<3$, there exist constants $c_{\ref{osngeneral}.3}$, $c_{\ref{osngeneral}.4}$ and $c_{\ref{osngeneral}.5}$ such that for $m \geq \max \{ (1+\theta)d, c_{\ref{osngeneral}.3}\log(4/\delta))/\theta^2 \}$ and $pm>c_{\ref{osngeneral}.4}{(\log (d/\delta ))^4}/{\theta^6}$, 

\begin{align*}
    \Pb \left( \kappa(\frac{1}{\sqrt{pm}}SU) \leq \frac{c_{\ref{osngeneral}.5}}{\theta} \right) \geq 1-\delta
\end{align*}

\end{theorem}

The proof of Theorem \ref{osngeneral} follows by applying the universality result \ref{thm:brailovskayavanhandel} to an augmented and symmetrized version of $SU$. More precisely, following \cite[p.443]{bandeira2023matrix}, for a $m \times d$ matrix $Y$, we define the augmented and symmetrized version of $Y$ as
\begin{align*}
\augsym(Y,\lambda)=\left[ {\begin{array}{*{20}{c}}
  {{0_{d \times d}}}&{{0_{d \times m}}}&{{Y^T}}&\aug(Y,\lambda)^{1/2} \\ 
  {{0_{m \times d}}}&{{0_{m \times m}}}&{{0_{m \times m}}}&{{0_{m \times d}}} \\ 
  Y&{{0_{m \times m}}}&{{0_{m \times m}}}&{{0_{m \times d}}} \\ 
  \aug(Y,\lambda)^{1/2}&{{0_{d \times m}}}&{{0_{d \times m}}}&{{0_{d \times d}}} 
\end{array}} \right]
\end{align*}
where
\begin{align*}
\aug(Y,\lambda) = (\norm{\E Y^T Y} +4 \lambda^2) \cdot \Id -\E Y^T Y
\end{align*} 
Then we set $X_{\lambda}=\augsym(Y,\lambda)$.

There are two reasons for using the matrix $\augsym(SU,\lambda)$. First, using Lemma \ref{lem:entries} requires symmetric matrices, so we need to symmetrize the matrix $SU$. Second, after symmetrization, we obtain the matrix
\begin{align*}
\left[ {\begin{array}{*{20}{c}}
  {}&{{{(SU)}^T}} \\ 
  {SU}&{} 
\end{array}} \right]
\end{align*}
and the spectrum of this matrix is the union of $\spec(SU)$ and $\{0\}$. By universality results, we can only claim that $\spec(SU) \cup \{0\}$ is close to $\spec(\sqrt{p}G) \cup \{0\}$, which does not directly imply the desired result that $s_{\min}(SU)$ is close to $s_{\min}(\sqrt{p}G)$. Therefore, we need to introduce the perturbations $\aug(SU,\lambda)$ to show that $s_{\min}(SU)$ is not close to zero.

To use Lemma \ref{thm:brailovskayavanhandel} we need to find out the corresponding guassian model and bound the parameters $\sigma(X_{\lambda}), \sigma_*(X_{\lambda})$ and $R(X_{\lambda})$ defined in Lemma \ref{thm:brailovskayavanhandel}.

By Lemma \ref{lem:entries}, we know that, in all of OSE-IID-ENT, OSNAP-IND-COL, and OSE-IND-DIAG distributions, each entry of $SU$ has variance $p$ and different entries have zero covariances. Therefore, we know that the corresponding gaussian model for $X_{\lambda}$ is the following.
\begin{align*}
\augsym(\sqrt{p}G,\lambda)=\left[ {\begin{array}{*{20}{c}}
  0&0&\sqrt{p}G^T&{{\aug(\sqrt{p}G,\lambda) }^{1/2}} \\ 
  0&0&0&0 \\ 
  {\sqrt{p}G}&0&0&0 \\ 
  {{\aug(\sqrt{p}G,\lambda) }^{1/2}}&0&0&0 
\end{array}} \right]
\end{align*}
where $G$ is an $m$ by $d$ i.i.d. standard gaussian random matrix.

Also, by the covariance structure of $SU$, we know that $\E[U^TS^TSU]=pm \cdot \Id$, and therefore
\begin{align*}
\aug(SU,\lambda) =& (\norm{\E (SU)^T (SU)} +4 \lambda^2) \cdot \Id -\E (SU)^T (SU)
\\=& (pm +4 \lambda^2) \cdot \Id - pm \cdot \Id
\\=& 4 \lambda^2 \cdot \Id
\end{align*}

Similarly, we also have 
\begin{align*}
\aug(\sqrt{p}G,\lambda) =& 4 \lambda^2 \cdot \Id
\end{align*}

We observe that $\sigma(X_{\lambda})$ and $ \sigma_*(X_{\lambda})$ do not depend on the decomposition of $X_{\lambda}$ as a sum of independent random matrices and can be calculated explicitly using the covariance structure of $X_{\lambda}$. Using this idea, we derive the following lemma that bounds $\sigma(X_{\lambda})$ and $ \sigma_*(X_{\lambda})$.
\begin{lemma}[Covariance Parameters]\label{mapa}
Let $S=\{s_{ij}\}_{i \in [m], j \in [n]}$ be a $m \times n$ random matrix such that $\E(s_{ij})=0$ and $\operatorname{Var}(s_{ij})=p$ for all $i \in [m], j \in [n]$, and $\cov(s_{ij},s_{kl})=0$ for any $\{i,k\} \subset [m], \{j,l\} \subset [n]$ and $ (i,j)\neq (k,l) $. Let $\sigma^*: L_{\infty}(\R) \otimes M_{2(m+d)}(\R) \to \R$ and $\sigma : L_{\infty}(\R) \otimes M_{2(m+d)}(\R) \to \R$ be the functions defined in Lemma \ref{thm:brailovskayavanhandel}. Then for any $\lambda>0$, we have
\begin{align*}
\sigma_*(\augsym(SU,\lambda)) \leq 2\sqrt{p} \text{ and } \sigma(\augsym(SU,\lambda)) \leq \sqrt{pm}
\end{align*}
\end{lemma}

\begin{proof}

Let $X_{\lambda}=\augsym(SU)$. Recall that $\sigma_*(X_{\lambda}) = \sup \limits _{\|v\|=\|w\|=1} \E[ |\langle v,(X_{\lambda}-\E X_{\lambda})w\rangle|^2]^{\frac{1}{2}}$.

We observe that
\begin{align*}
\E X_{\lambda}=&\E \left[ {\begin{array}{*{20}{c}}
  0&0&(SU)^T&{{(4 \lambda^2 \cdot \Id)}^{1/2}} \\ 
  0&0&0&0 \\ 
  {(SU)}&0&0&0 \\ 
  {{(4 \lambda^2 \cdot \Id)}^{1/2}}&0&0&0 
\end{array}} \right]
\\=&\left[ {\begin{array}{*{20}{c}}
  0&0&0&{{(4 \lambda^2 \cdot \Id)}^{1/2}} \\ 
  0&0&0&0 \\ 
  {0}&0&0&0 \\ 
  {{(4 \lambda^2 \cdot \Id)}^{1/2}}&0&0&0 
\end{array}} \right]
\end{align*}

Therefore, we have
\begin{align*}
X_{\lambda}-\E X_{\lambda}=\left[ {\begin{array}{*{20}{c}}
  0&0&(SU)^T&{0} \\ 
  0&0&0&0 \\ 
  {(SU)}&0&0&0 \\ 
  {0}&0&0&0 
\end{array}} \right]
\end{align*}

Fixing vectors $v,w$, we have
\begin{align*}
\langle v,(X_{\lambda}-\E X_{\lambda})w\rangle =&\sum\limits_{1 \le i \le 2(d + m)} {\sum\limits_{1 \le j \le 2(d + m)} {(X_{\lambda}-\E X_{\lambda})_{i,j}} \cdot {w_j}{v_i}}  \\
=&\sum\limits_{1 \le i \le m } \sum\limits_{1 \le j \le  d}(X_{\lambda}-\E X_{\lambda})_{m+d+i,j}{w_j}{v_{m+d+i}}\\&+\sum\limits_{1 \le i \le d } \sum\limits_{1 \le j \le  m}(X_{\lambda}-\E X_{\lambda})_{i,m+d+j}{w_{m+d+j}}{v_{i}}
\\=&\sum\limits_{1 \le i \le m } \sum\limits_{1 \le j \le  d}(SU)_{i,j}{w_j}{v_{m+d+i}}+\sum\limits_{1 \le i \le d } \sum\limits_{1 \le j \le  m}(SU)_{j,i}{w_{m+d+j}}{v_{i}}
\\=&\langle v_{[m+d+1:m+d+m]}, SUw_{[1:d]} \rangle + \langle w_{[m+d+1:m+d+m]}, SUv_{[1:d]} \rangle
\end{align*}   
Let $\eta=Uw_{[1:d]}$, $\xi=Uv_{[1:d]}$, then $\norm{\eta}_2 \le 1$ and $\norm{\xi}_2 \le 1$. Also, we have 
\begin{align*}
    \langle v,(X_{\lambda}-\E X_{\lambda})w\rangle &= \langle v_{[m+d+1:m+d+m]}, S\eta \rangle + \langle w_{[m+d+1:m+d+m]}, S\xi \rangle   \\
    &= \sum\limits_{1 \le i \le m } \sum\limits_{1 \le j \le  n}s_{i,j}(v_{m+d+i}\eta_j + w_{m+d+i}\xi_j)
\end{align*}

Using the condition that $\E(s_{ij})=0$ and $\operatorname{Var}(s_{ij})=p$ for all $i \in [m], j \in [n]$, and $\cov(s_{ij},s_{kl})=0$ for any $i \in [m], j \in [n]$, we conclude

\begin{align*}
    \E[ |\langle v,(X_{\lambda}-\E X_{\lambda})w\rangle|^2] \leq & p(\sum\limits_{1 \le i \le m } \sum\limits_{1 \le j \le  n}(v_{m+d+i}\eta_j + w_{m+d+i}\xi_j)^2) \\
    \leq& p(\sum\limits_{1 \le i \le m } \sum\limits_{1 \le j \le  n} v_{m+d+i}^2\eta_j^2 + w_{m+d+i}^2\xi_j^2 + 2v_{m+d+i}\eta_jw_{m+d+i}\xi_j)\\
    \leq& p(\sum\limits_{1 \le i \le m } v_{m+d+i}^2\sum\limits_{1 \le j \le  n} \eta_j^2 + \sum\limits_{1 \le i \le m }w_{m+d+i}^2 \sum\limits_{1 \le j \le  n}\xi_j^2 \\& + 2\sum\limits_{1 \le i \le m } \sum\limits_{1 \le j \le  n}v_{m+d+i}\eta_jw_{m+d+i}\xi_j)        
        \\
    = & p(\norm{v_{[m+d+1:m+d+m]}}_2^2\norm{\eta}_2^2 + \norm{w_{[m+d+1:m+d+m]}}_2^2\norm{\xi}_2^2 \\&+ 2\langle v_{[m+d+1:m+d+m]}, w_{[m+d+1:m+d+m]} \rangle \langle \eta, \xi \rangle ) \\ \leq & 4p
\end{align*}
where we use Schwarz inequality to bound the term $\langle v_{[m+d+1:m+d+m]}, w_{[m+d+1:m+d+m]} \rangle \langle \eta, \xi \rangle$.

Therefore, we conclude that $\sigma_*(X_{\lambda}) \leq 2\sqrt{p}$.

To calculate $\sigma(X_{\lambda}) $, we recall the definition
\begin{align*}
\sigma(X_{\lambda}) = \norm{\E [(X_{\lambda}-\E X_{\lambda})^2]}_{op}^{1/2}
\end{align*}

By our previous calculation, we have
\begin{align*}
X_{\lambda}-\E X_{\lambda}=\left[ {\begin{array}{*{20}{c}}
  0&0&(SU)^T&{0} \\ 
  0&0&0&0 \\ 
  {(SU)}&0&0&0 \\ 
  {0}&0&0&0 
\end{array}} \right]
\end{align*}
and therefore we have
\begin{align*}
(X_{\lambda}-\E X_{\lambda})^2=&\left[ {\begin{array}{*{20}{c}}
  (SU)^T{(SU)}&0&0&{0} \\ 
  0&0&0&0 \\ 
  0&0&{(SU)}(SU)^T&0 \\ 
  {0}&0&0&0 
\end{array}} \right]
\\=&\left[ {\begin{array}{*{20}{c}}
  U^TS^T{(SU)}&0&0&{0} \\ 
  0&0&0&0 \\ 
  0&0&{(SU)}U^TS^T&0 \\ 
  {0}&0&0&0 
\end{array}} \right]
\end{align*}

We can easily calculate $\E[SUU^TS^T] = pd\cdot I$ and $\E[U^TS^TSU]=pm \cdot I$ from the covariance structure of $SU$. Therefore, we conclude that $\sigma(X_{\lambda}) \leq \sqrt{pm}$.

\end{proof}

\begin{proof}[Proof of Theorem \ref{osngeneral}]

Using Lemma \ref{mapa} and Lemma \ref{lem:entries}, we have \begin{align*}\sigma_*(\augsym(SU,\lambda)) \leq 2\sqrt{p}\end{align*} and \begin{align*}\sigma(\augsym(SU,\lambda)) \leq \sqrt{pm}\end{align*} for all the three distributions.

$R(X_{\lambda})$ depends on the decomposition of $X_{\lambda}$ as a sum of independent random matrices, so we write the matrix $SU$ as a sum of independent random matrices in each of the three distributions as follows.

For the OSE-IID-ENT distribution, we observe that $S=\sum\limits_{i,j} {{s_{i,j}}({e_i}{e_j}^T)} $, so we have $SU = \sum\limits_{i,j} {{s_{i,j}}({e_i}{e_j}^T)U}  = \sum\limits_{i,j} {{s_{i,j}}({e_i}{u_j}^T)} $, where ${u_j}$ is the $j$th row vector of the matrix $U$. Therefore, we have
    \begin{align*}
    \augsym(SU,\lambda)=& \left[ {\begin{array}{*{20}{c}}
      0&0&(SU)^T&{{(4 \lambda^2 \cdot \Id)}^{1/2}} \\ 
      0&0&0&0 \\ 
      {(SU)}&0&0&0 \\ 
      {{(4 \lambda^2 \cdot \Id)}^{1/2}}&0&0&0 
    \end{array}} \right]
    \\=&\left[ {\begin{array}{*{20}{c}}
      0&0&0&{{(4 \lambda^2 \cdot \Id)}^{1/2}} \\ 
      0&0&0&0 \\ 
      {0}&0&0&0 \\ 
      {{(4 \lambda^2 \cdot \Id)}^{1/2}}&0&0&0 
    \end{array}} \right] \\&+\sum\limits_{i,j} {{s_{i,j}}\left[ {\begin{array}{*{20}{c}}
          0&0&({e_i}{u_j}^T)^T&0 \\ 
          0&0&0&0 \\ 
          {({e_i}{u_j}^T)}&0&0&0 \\ 
          0&0&0&0 
        \end{array}} \right]} 
    \end{align*}

Since $s_{i,j}$ is bounded by $1$, we have
\begin{align*}
        R(\augsym(SU,\lambda)) \le \max \limits _{i,j} \left\| {\left[ {\begin{array}{*{20}{c}}
          0&0&{{{({e_i}{u_j}^T)}^T}}&0 \\ 
          0&0&0&0 \\ 
          {({e_i}{u_j}^T)}&0&0&0 \\ 
          0&0&0&0 
        \end{array}} \right]} \right\|_{op} \le 1
\end{align*}

For the OSE-IND-DIAG distribution, we have $SU = \sum_{l=1}^{pn} Y_lU$, where $Y_l=F_{\gamma_l}(W_l)$ as in Definition \ref{oseinddiag}. So we have
    \begin{align*}
        \augsym(SU,\lambda)=& \left[ {\begin{array}{*{20}{c}}
          0&0&(SU)^T&{{(4 \lambda^2 \cdot \Id)}^{1/2}} \\ 
          0&0&0&0 \\ 
          {(SU)}&0&0&0 \\ 
          {{(4 \lambda^2 \cdot \Id)}^{1/2}}&0&0&0 
        \end{array}} \right]
        \\=&\left[ {\begin{array}{*{20}{c}}
          0&0&0&{{2 \lambda \cdot \Id}} \\ 
          0&0&0&0 \\ 
          {0}&0&0&0 \\ 
          {{2 \lambda \cdot \Id}}&0&0&0 
        \end{array}} \right] \\&+\sum_{l=1}^{pn}\left[ {\begin{array}{*{20}{c}}
              0&0&(Y_lU)^T&0 \\ 
              0&0&0&0 \\ 
              {(Y_lU)}&0&0&0 \\ 
              0&0&0&0 
            \end{array}} \right]
        \end{align*}
Since $\norm{Y_lU}_{op} \leq \norm{Y_l}_{op}\norm{U}_{op} \leq 1$, we have
\begin{align*}
        R(\augsym(SU,\lambda)) \le \max \limits _{l} \left\| {\begin{array}{*{20}{c}}
                      0&0&(Y_lU)^T&0 \\ 
                      0&0&0&0 \\ 
                      {(Y_lU)}&0&0&0 \\ 
                      0&0&0&0 
                    \end{array}} \right\|_{op} \le 1
\end{align*}

For the OSNAP-IND-COL distribution, the sum is similar to the OSE-IND-DIAG case. More precisely, for $k\in [s], l \in [n]$, we define $Y_{k,l}$ to be the $m \times n$ matrix such that $(Y_{k,l})_{i,j}=\mathbbm{1}_{[(m/s)(i-1)+1:(m/s)i]\times\{ j\}}((i,j))S_{i,j}$. By the construction of OSNAP-IND-COL distribution, the matrix $Y_{k,l}$ has at most one nonzero entry, so we still have $\norm{Y_{k,l}U}_{op} \leq \norm{Y_{k,l}}_{op}\norm{U}_{op} \leq 1$. And then we have
\begin{align*}
        R(\augsym(SU,\lambda)) \le \max \limits _{k,l} \left\| {\begin{array}{*{20}{c}}
                      0&0&(Y_{k,l}U)^T&0 \\ 
                      0&0&0&0 \\ 
                      {(Y_{k,l}U)}&0&0&0 \\ 
                      0&0&0&0 
                    \end{array}} \right\|_{op} \le 1
\end{align*}

In conclusion, we have $R(\augsym(SU,\lambda)) \le 1$ for all the three distributions.

Using lemma \ref{thm:brailovskayavanhandel} with $t=\log(2d/\delta)$, we have
\begin{align*}
\Pb(d_H(\spec(\augsym(SU)),\spec(\augsym(\sqrt{p}G)))>c_1\zeta(\log(2d/\delta)))\leq\delta/2
\end{align*}
where
\begin{align*}
\zeta(t) =
	\sigma_*(X_{\lambda})  t^{\frac{1}{2}} +
	R(X_{\lambda})^{\frac{1}{3}}\sigma(X_{\lambda})^{\frac{2}{3}} t^{\frac{2}{3}} +
	R(X_{\lambda}) t
\end{align*}
for some constant $c_1$. Without loss of generality, we can assume that $c_1>1$.

Using lemma \ref{lem:Gaussianspectrum}, we have
\begin{multline*}
     \Pb \bigg( \sqrt{pm} - \sqrt{pd} - \sqrt{2p\log(4/\delta)} \leq s_{\min}(\sqrt{p}G) \\  \leq s_{\max}(\sqrt{p}G)    \sqrt{pm} + \sqrt{pd} + \sqrt{2p\log(4/\delta)} \bigg) \geq 1 - \delta/2
\end{multline*}

Let $\mathcal{E}$ be the event
\begin{align*}
\mathcal{E}= &\{ \sqrt{pm} - \sqrt{pd} - \sqrt{2p\log(4/\delta)} \leq s_{\min}(\sqrt{p}G) \\  &\leq s_{\max}(\sqrt{p}G) \le   \sqrt{pm} + \sqrt{pd} + \sqrt{2p\log(4/\delta)} \bigg\} \\ &\cap \{d_H(\spec(\augsym(SU,\lambda)),\spec(\augsym(\sqrt{p}G,\lambda))) \le c_1\zeta(\log(2d/\delta))\}
\end{align*}

Therefore, we have
\begin{align*}
\Pb(\mathcal{E}) \ge 1- \delta
\end{align*}
by the union bound.

Assume that the event $\mathcal{E}$ happens, and we want to obtain desired bounds for the largest and smallest singular value for $SU$. To this end, we need the following lemma that connects the singular values of $\augsym(SU,\lambda)$ to the singular values of $SU$.

\begin{lemma}[Connections between Singular Values]\label{connectsing}
Assume that $\lambda \ge c_1\zeta(\log(2d/\delta))$. When the event
\begin{align*}
\mathcal{E}= &\{ \sqrt{pm} - \sqrt{pd} - \sqrt{2p\log(4/\delta)} \leq s_{\min}(\sqrt{p}G) \\  &\leq s_{\max} (\sqrt{p}G) \le   \sqrt{pm} + \sqrt{pd} + \sqrt{2p\log(4/\delta)} \bigg\} \\ &\cap \{d_H(\spec(\augsym(SU,\lambda)),\spec(\augsym(\sqrt{p}G,\lambda))) \le c_1\zeta(\log(2d/\delta))\}
\end{align*}
happens, we have
\begin{align*}
&\sqrt{pm} - \sqrt{pd} - \sqrt{2p\log(4/\delta)} - 5 \lambda \leq s_{\min}(SU) \\  &\leq s_{\max}(SU) \leq   \sqrt{pm} + \sqrt{pd} + \sqrt{2p\log(4/\delta)} + 5 \lambda
\end{align*}
\end{lemma}

\begin{proof} We use the following observations to connect the singular values of $SU$ with \begin{align*}
\spec(\augsym(SU,\lambda))
\end{align*} and connect singular values of $\sqrt{p}G$ with \begin{align*}
\spec(\augsym(\sqrt{p}G,\lambda))
\end{align*}.

It follows from SVD that $s_{\min}(SU)^2=s_{\min}((SU)^T(SU))$ and $s_{\max}(SU)^2=s_{\max}((SU)^T(SU))$.

By triangle inequality, for $\lambda>0$, we have
\begin{align*}
s_{\min}((SU)^T(SU)) \ge & s_{\min}((SU)^T(SU)+(2\lambda \Id)^2) -s_{\max}((2\lambda \Id)^2) \\=& s_{\min}((SU)^T(SU)+(2\lambda \Id)^2) -4 \lambda^2
\end{align*}
and
\begin{align*}
s_{\max}((SU)^T(SU)) \le & s_{\max}((SU)^T(SU)+(2\lambda \Id)^2) +s_{\max}((2\lambda \Id)^2) \\=& s_{\max}((SU)^T(SU)+(2\lambda \Id)^2) +4 \lambda^2
\end{align*}

By Remark 2.6 in \cite{bandeira2023matrix}, we know that 
\begin{align*}
\sqrt{s_{\min}((SU)^T(SU)+(2\lambda \Id)^2)} \in \spec(\augsym(SU,\lambda)) \cup \{0\}
\end{align*}
and
\begin{align*}
\sqrt{s_{\max}((SU)^T(SU)+(2\lambda \Id)^2)} \in \spec(\augsym(SU,\lambda)) \cup \{0\}
\end{align*}

Similarly, we have
\begin{align*}
s_{\min}((\sqrt{p}G)^T(\sqrt{p}G)+(2\lambda \Id)^2) \ge & s_{\min}((\sqrt{p}G)^T(\sqrt{p}G)) -s_{\max}((2\lambda \Id)^2) \\ =& s_{\min}((\sqrt{p}G)^T(\sqrt{p}G)) -4 \lambda^2
\\ =& (\sqrt{pm} - \sqrt{pd} - \sqrt{2p\log(4/\delta)})^2 -4 \lambda^2
\end{align*}
and
\begin{align*}
s_{\max}((\sqrt{p}G)^T(\sqrt{p}G)+(2\lambda \Id)^2) \le & s_{\max}((\sqrt{p}G)^T(\sqrt{p}G)) +s_{\max}((2\lambda \Id)^2) \\=& s_{\max}((\sqrt{p}G)^T(\sqrt{p}G)) +4 \lambda^2
\\=& (\sqrt{pm} + \sqrt{pd} + \sqrt{2p\log(4/\delta)})^2 +4 \lambda^2
\end{align*}

Therefore, we know that
\begin{align*}
&\spec((\sqrt{p}G)^T(\sqrt{p}G)+(2\lambda \Id)^2) \\ \subset & [(\sqrt{pm} - \sqrt{pd} - \sqrt{2p\log(4/\delta)})^2 -4 \lambda^2, (\sqrt{pm} + \sqrt{pd} + \sqrt{2p\log(4/\delta)})^2 +4 \lambda^2]
\end{align*}

Using this result with the proof of Lemma 4.9 in \cite{bandeira2023matrix}, we know that
\begin{align*}
&\spec(\augsym(\sqrt{p}G,\lambda)^2) \cup \{0\}
\\=& \spec(\left[ {\begin{array}{*{20}{c}}
  {{{(\sqrt p G)}^T}}&{2\lambda \Id} \\ 
  {{0_{m \times m}}}&{{0_{m \times d}}} 
\end{array}} \right] \left[ {\begin{array}{*{20}{c}}
  {{{(\sqrt p G)}^T}}&{2\lambda \Id} \\ 
  {{0_{m \times m}}}&{{0_{m \times d}}} 
\end{array}} \right]^T) \cup \{0\}
\\ \subset & [(\sqrt{pm} - \sqrt{pd} - \sqrt{2p\log(4/\delta)})^2 -4 \lambda^2, (\sqrt{pm} + \sqrt{pd} + \sqrt{2p\log(4/\delta)})^2 +4 \lambda^2] \cup \{0\}
\end{align*}

Therefore, we have
\begin{align*}
&\spec(\augsym(\sqrt{p}G,\lambda)) \cup \{0\}
\\ \subset & [\sqrt{(\sqrt{pm} - \sqrt{pd} - \sqrt{2p\log(4/\delta)})^2 -4 \lambda^2}, \sqrt{(\sqrt{pm} + \sqrt{pd} + \sqrt{2p\log(4/\delta)})^2 +4 \lambda^2}] \cup \{0\} 
\\ &  \cup -[\sqrt{(\sqrt{pm} - \sqrt{pd} - \sqrt{2p\log(4/\delta)})^2 -4 \lambda^2}, \sqrt{(\sqrt{pm} + \sqrt{pd} + \sqrt{2p\log(4/\delta)})^2 +4 \lambda^2}]
\end{align*}

Assume that $\lambda \ge c_1\zeta(\log(2d/\delta))$. Then using the fact that
\begin{align*}
\sqrt{s_{\max}((SU)^T(SU)+(2\lambda \Id)^2)} \in & \spec(\augsym(SU,\lambda)) \cup \{0\}
\end{align*}
and
\begin{align*}
\spec(\augsym(SU,\lambda)) \cup \{0\}
\subset (\spec(\augsym(\sqrt{p}G,\lambda))+\lambda) \cup \{0\}
\end{align*}
we have
\begin{align*}
\sqrt{s_{\max}((SU)^T(SU)+(2\lambda \Id)^2)} \le \sqrt{(\sqrt{pm} + \sqrt{pd} + \sqrt{2p\log(4/\delta)})^2 +4 \lambda^2}+\lambda
\end{align*}

Therefore, we have
\begin{align*}
s_{\max}((SU)^T(SU)) \le & s_{\max}((SU)^T(SU)+(2\lambda \Id)^2) +4 \lambda^2
\\ \le & (\sqrt{(\sqrt{pm} + \sqrt{pd} + \sqrt{2p\log(4/\delta)})^2 +4 \lambda^2}+\lambda)^2 + 4 \lambda^2
\\ \le & (\sqrt{pm} + \sqrt{pd} + \sqrt{2p\log(4/\delta)}+ 2 \lambda +\lambda +2 \lambda)^2 
\end{align*}
where we use $a^2+b^2 \le (a+b)^2$ for $a \ge 0$ and $b \ge 0$.

It follows directly that
\begin{align*}
s_{\max}(SU)=&\sqrt{s_{\max}((SU)^T(SU))}
\\ \le &  \sqrt{pm} + \sqrt{pd} + \sqrt{2p\log(4/\delta)}+ 5 \lambda
\end{align*}

Similarly, we have
\begin{align*}
\sqrt{s_{\min}((SU)^T(SU)+(2\lambda \Id)^2)} \in (\spec(\augsym(\sqrt{p}G,\lambda))+\lambda) \cup \{0\}
\end{align*}
and therefore we know that either 
\begin{align*}
\sqrt{s_{\min}((SU)^T(SU)+(2\lambda \Id)^2)} \in [0,\lambda]
\end{align*}
or
\begin{align*}
&\sqrt{s_{\min}((SU)^T(SU)+(2\lambda \Id)^2)} \\ \in& [\sqrt{(\sqrt{pm} - \sqrt{pd} - \sqrt{2p\log(4/\delta)})^2 -4 \lambda^2}-\lambda, \sqrt{(\sqrt{pm} + \sqrt{pd} + \sqrt{2p\log(4/\delta)})^2 +4 \lambda^2}+\lambda]
\end{align*}

We can exclude the first possibility by the argument below. Using the representation, e.g., in \cite[Example 7.5.1]{meyer2023matrix}
\begin{align*}
s_{\min}((SU)^T(SU)+(2\lambda \Id)^2)= \min_{\norm{x}_2=1} \ip{x,((SU)^T(SU)+(2\lambda \Id)^2)x}
\end{align*}
and the fact that $(SU)^T(SU)$ is positive semi-definite, we know that
\begin{align*}
s_{\min}((SU)^T(SU)+(2\lambda \Id)^2) \ge 4 \lambda ^2
\end{align*}
which implies
\begin{align*}
\sqrt{s_{\min}((SU)^T(SU)+(2\lambda \Id)^2)} \notin [0,\lambda]
\end{align*}

Therefore, we have
\begin{align*}
&\sqrt{s_{\min}((SU)^T(SU)+(2\lambda \Id)^2)} \\ \in& [\sqrt{(\sqrt{pm} - \sqrt{pd} - \sqrt{2p\log(4/\delta)})^2 -4 \lambda^2}-\lambda, \sqrt{(\sqrt{pm} + \sqrt{pd} + \sqrt{2p\log(4/\delta)})^2 +4 \lambda^2}+\lambda]
\end{align*}
which means
\begin{align*}
&\sqrt{s_{\min}((SU)^T(SU)+(2\lambda \Id)^2)} \\ \ge& \sqrt{(\sqrt{pm} - \sqrt{pd} - \sqrt{2p\log(4/\delta)})^2 -4 \lambda^2}-\lambda
\end{align*}

Therefore, we have
\begin{align*}
s_{\min}((SU)^T(SU)) \ge & s_{\min}((SU)^T(SU)+(2\lambda \Id)^2) -4 \lambda^2
\\ \ge & (\sqrt{(\sqrt{pm} - \sqrt{pd} - \sqrt{2p\log(4/\delta)})^2 -4 \lambda^2}-\lambda)^2 - 4 \lambda^2
\end{align*}

For $a>b>0$, by $(a-b)^2+b^2 \le (a-b+b)^2$, we have $a^2-b^2 \ge (a-b)^2$.

Therefore, we have
\begin{align*}
s_{\min}(SU) =& \sqrt{s_{\min}((SU)^T(SU))} \\ \ge & \sqrt{(\sqrt{(\sqrt{pm} - \sqrt{pd} - \sqrt{2p\log(4/\delta)})^2 -4 \lambda^2}-\lambda)^2 - 4 \lambda^2}
\\ \ge & \sqrt{pm} - \sqrt{pd} - \sqrt{2p\log(4/\delta)}-2 \lambda -\lambda -2 \lambda 
\\ = & \sqrt{pm} - \sqrt{pd} - \sqrt{2p\log(4/\delta)}-5 \lambda 
\end{align*}

In conclusion, if the event $\mathcal{E}$ happens, we always have
\begin{align*}
&\sqrt{pm} - \sqrt{pd} - \sqrt{2p\log(4/\delta)} - 5 \lambda \leq s_{\min}(SU) \\  &\leq s_{\max}(SU) \leq   \sqrt{pm} + \sqrt{pd} + \sqrt{2p\log(4/\delta)} + 5 \lambda
\end{align*}

\end{proof}

From Lemma \ref{connectsing}, we derive that
\begin{align*}
     \Pb \bigg( &\sqrt{pm} - \sqrt{pd} - \sqrt{2p\log(4/\delta)} - 5 \lambda \leq s_{\min}(SU) \\  &\leq s_{\max}(SU) \leq   \sqrt{pm} + \sqrt{pd} + \sqrt{2p\log(4/\delta)} + 5 \lambda \bigg) \geq 1 - \delta
\end{align*}

Now, we choose $\lambda = \frac{1}{10} \varepsilon \sqrt{pm}$. This is possible when $\frac{1}{10} \varepsilon \sqrt{pm} \ge c_1\zeta(\log(2d/\delta))$, and we will simplify this condition later.

Assuming that we can choose $\lambda = \frac{1}{10} \varepsilon \sqrt{pm}$, we have
\begin{align*}
     \Pb \bigg( &\sqrt{pm} - \sqrt{pd} - \sqrt{2p\log(4/\delta)} - \frac{1}{2} \varepsilon \sqrt{pm}  \leq s_{\min}(SU) \\  &\leq s_{\max}(SU) \leq   \sqrt{pm} + \sqrt{pd} + \sqrt{2p\log(4/\delta)} + \frac{1}{2} \varepsilon \sqrt{pm} \bigg) \geq 1 - \delta
\end{align*}

Assuming also that $m > \max \{ \frac{16d}{\varepsilon^2}, \frac{32\log(4/\delta)}{\varepsilon^2} \}$, we have $\frac{\sqrt{2p\log(4/\delta)}}{\sqrt{pm}}<\frac{\varepsilon}{4}$ and $\sqrt{\frac{d}{m}}<\frac{\varepsilon}{4}$, which lead to
\begin{align*}
\Pb \bigg(& 1 - \frac{\varepsilon}{4} - \frac{\varepsilon}{4} - \frac{\varepsilon}{2} \leq s_{\min}((1/\sqrt{pm})SU) \\   &\leq s_{\max}((1/\sqrt{pm})SU) \leq 1 + \frac{\varepsilon}{4} + \frac{\varepsilon}{4} + \frac{\varepsilon}{2} \bigg) \geq 1-\delta
\end{align*}

In other words, we have
\begin{align*}
\Pb \bigg( 1 - \varepsilon \leq s_{\min}((1/\sqrt{pm})SU)   \leq s_{\max}((1/\sqrt{pm})SU) \leq 1 + \varepsilon \bigg) \geq 1-\delta
\end{align*}
which is exactly what we want.

Then it suffices to translate the condition $\frac{1}{10} \varepsilon \sqrt{pm} \ge c_1\zeta(\log(2d/\delta))$ into the requirements for $p$, $m$, and $d$. To this end, we first calculate that
\begin{align*}
\zeta(t)=2\sqrt{p}t^{1/2}+(\sqrt{pm})^{2/3}t^{2/3}+t
\end{align*}
by the earlier bounds for $\sigma(X), \sigma_*(X)$ and $R(X)$.

We claim that it is enough to require that,
\begin{align*}
     (\sqrt{pm})^{2/3}(\log(2d/\delta))^{2/3} \leq \frac{\varepsilon}{40c_1}\sqrt{pm} 
\end{align*}

Equivalently, we just need
\begin{align*}
    \frac{({40c_1})^6(\log(2d/\delta))^4}{\varepsilon^6} \leq pm
\end{align*}

In fact, if $pm \geq c_2\frac{(\log(2d/\delta))^4}{\varepsilon^6}$ where $c_2=({40c_1})^6$, we will also have $\log(2d/\delta) \leq  \frac{\varepsilon}{40c_1}\sqrt{pm}$ and $2(\log(2d/\delta))^{1/2} \leq \frac{\varepsilon}{20c_1}\sqrt{pm}$, which gives us $c_1\zeta(\log(2d/\delta)) \leq \frac{\varepsilon}{10}\sqrt{pm}$, and therefore we have proved the first part of the theorem.

To prove the second part, we observe that, when $m\geq (1+\theta)d$, 
\begin{multline*}
     \Pb \bigg( \sqrt{pm} - \sqrt{pd} - \sqrt{2p\log(4/\delta)} - 5 \lambda \leq s_{\min}(SU) \\  \leq s_{\max}(SU) \leq   \sqrt{pm} + \sqrt{pd} + \sqrt{2p\log(4/\delta)} + 5 \lambda \bigg) \geq 1 - \delta
\end{multline*}
implies
\begin{multline*}
    \Pb \bigg( 1 -\frac{1}{\sqrt{1+\theta}} - \frac{\sqrt{2p\log(4/\delta)}}{\sqrt{pm}} - \frac{5 \lambda}{\sqrt{pm}} \leq s_{\min}((1/\sqrt{pm})SU) \\   \leq s_{\max}((1/\sqrt{pm})SU) \leq 1 + \frac{1}{\sqrt{1+\theta}} + \frac{\sqrt{2p\log(4/\delta)}}{\sqrt{pm}} + \frac{5 \lambda}{\sqrt{pm}} \bigg) \geq 1-\delta
\end{multline*}

Now for $\theta<3$, $\frac{1}{\sqrt{1+\theta}} < 1- \frac{\theta}{6}$, and for $m\geq \frac{36\cdot16\cdot2\cdot \log(4/\delta)}{\theta^2}$, we get,
\begin{multline*}
    \Pb \bigg( \frac{\theta}{6} - \frac{\theta}{24} - \frac{5 \lambda}{\sqrt{pm}} \leq s_{\min}((1/\sqrt{pm})SU)   \leq s_{\max}((1/\sqrt{pm})SU) \leq 3 + \frac{5 \lambda}{\sqrt{pm}} \bigg) \geq 1-\delta
\end{multline*}

Repeating the previous analysis now shows that $c_1\zeta(\log(2d/\delta)) \leq \frac{\theta}{24 \cdot 5}\sqrt{pm}$ when $pm \geq c_3\frac{(\log(2d/\delta))^4}{\theta^6}$ for some large enough constant $c_3$, so we can choose $\lambda = \frac{\theta}{24 \cdot 5}\sqrt{pm}$. In this case, we have
\begin{align*}
    \Pb \bigg( \frac{\theta}{12} \leq s_{\min}((1/\sqrt{pm})SU)    \leq s_{\max}((1/\sqrt{pm})SU) \leq 4 \bigg) \geq 1-\delta
\end{align*}
\end{proof}

%% file: nonoblivious.tex
\section{Analysis of Non-oblivious Sparse Embeddings}

In the non-oblivious setting, we are looking to embed a specific $d$-dimensional subspace of $\R^n$ represented by a matrix $U$, and we assume that we have access to $(\beta_1,\beta_2)$-approximate leverage scores of $U$ as defined in Section \ref{sec:prelim}. Given access to this information, we can modify the oblivious models of $S$ to give more weight to certain coordinates of the ambient space. We refer to this approach as Leverage Score Sparsitication (LESS).

\begin{figure}
    \centering
    \includegraphics[width=0.7\linewidth]{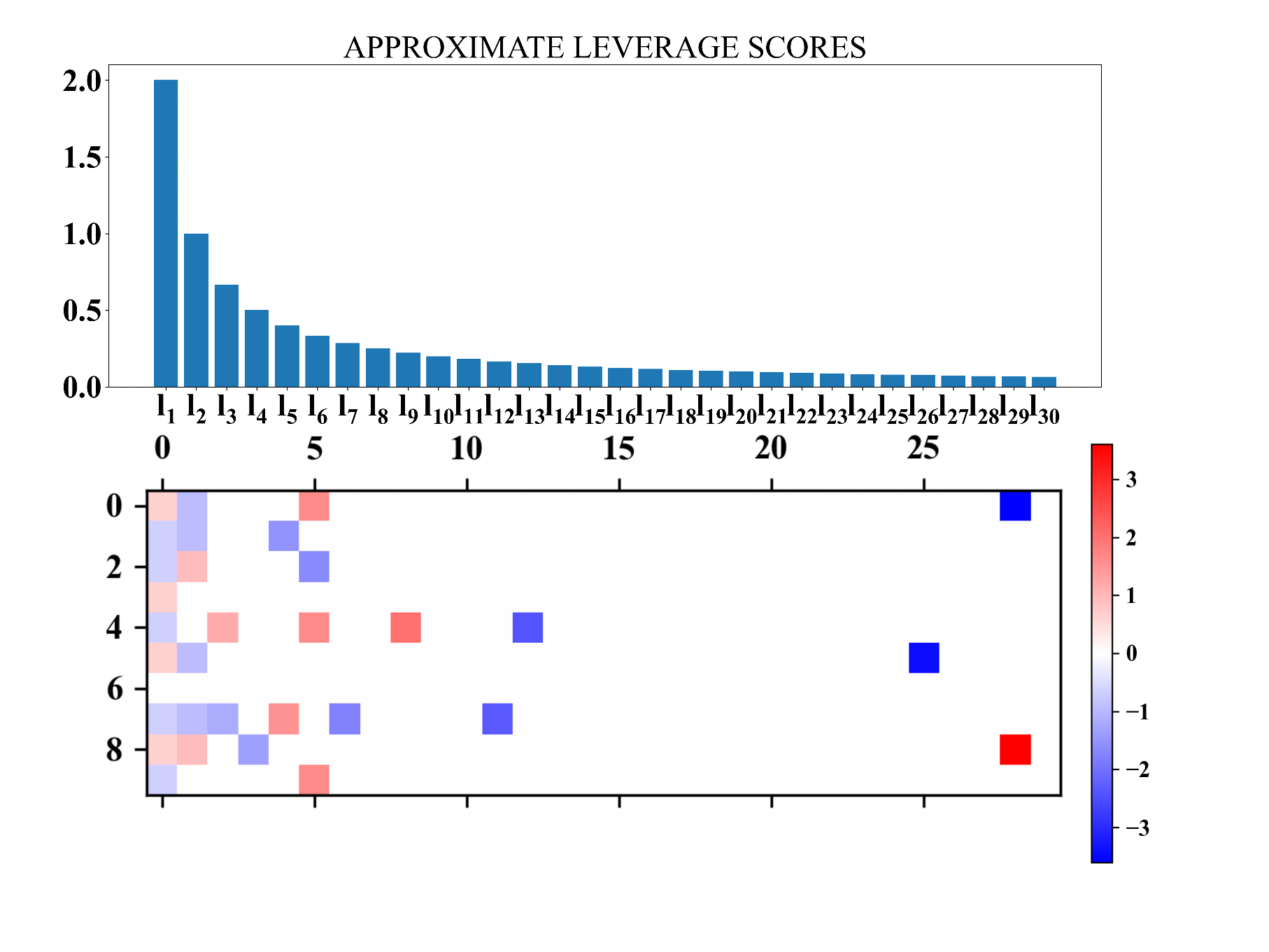}
    \caption{LESS-IND-ENT with decreasing leverage scores. Since the probability of an entry being non-zero is proportional to the corresponding leverage score, we see that the matrix becomes sparser as we move in the direction of decreasing leverage scores. Since the scaling of entries is inversely proportional to the square root of the corresponding leverage score, the magnitude of the non-zero entries becomes larger as we move to the right.}
    \label{fig:nseindent} 
\end{figure}

\subsection{Leverage Score Sparsification}

We propose two variants of a LESS embedding. First, we consider an extension of the OSE-IND-ENT model (with i.i.d. entries) studied in Section \ref{s:analysis-direct}, with the Bernoulli sparsifier for each entry of $S$ being non-zero with a probability proportional to the leverage score of the corresponding component. However, we still want the variance of each entry of $S$ to be $p$, so the entries are appropriately scaled copies of $\pm 1$ (See Figure \ref{fig:nseindent}). Similar to the previous section, the analysis proceeds by considering a symmetrized version of $SU$ where $S$ is viewed as a sum of independent random matrices $Z_i$. The scaling of entries of $S$ retains the bound on $\norm{Z_i U}$ as mentioned in Section \ref{s:analysis-direct}. A formal definition follows.

\begin{definition}[LESS-IND-ENT]\label{noseie}
An $m \times n$ random matrix $S$ is called a leverage score sparsified embedding with independent entries (LESS-IND-ENT) corresponding to $(\beta_1,\beta_2)$-approximate leverage scores $(l_1,...,l_n)$ with parameter $p$, if $S$ has entries $s_{i,j}=\frac{1}{\sqrt{\beta_1  l_j}} \delta_{i,j} \xi_{i,j}$ where $\delta_{i,j}$ are independent Bernoulli random variables taking value 1 with probability $p_{ij}= \beta_1 l_j p$ and $\xi_{i,j}$ are i.i.d. random variables with $\Pb(\xi_{i,j}=1)=\Pb(\xi_{i,j}=-1)=1/2$. 
\end{definition}

In the next variant of a LESS embedding, we are able to reduce the computational and random bit complexity for generating sparsity by only generating as many non-zero entries as required. Here, the $i\textsuperscript{th}$ row of $S$ is the sum of $np$ many i.i.d. random matrices $Z_{ij}$ where each $Z_{ij}$ is determined by choosing one entry from the $n$ possible entries of the $i\textsuperscript{th}$ row and setting the remaining entries to 0. Here, instead of choosing the positions uniformly at random, we choose them proportionally to the corresponding leverage score.

\begin{definition}[LESS-IND-ROWS]\label{def:lessindrows}
Assume that $(\beta_1  p  \sum l_i)$ is an integer. An $m \times n$ random matrix $S$ is called a leverage score sparsified embedding with independent rows (LESS-IND-ROWS) corresponding to  $(\beta_1,\beta_2)$-approximate leverage scores $(l_1,...,l_n)$ with parameter $p$, if the $i\textsuperscript{th}$ row of $S$ is a sum of $(\beta_1  p \sum l_j)$ i.i.d. copies $Z_{i1},Z_{i2},...,$ of a random variable $Z_i$, i.e.,
\begin{align*}S= \sum_{i=1}^m \sum \limits_{k=1}^{(\beta_1  p  \sum l_i)} Z_{ik},
\end{align*}
where $Z_{i}$ is defined as follows. Let $\gamma$ be a random variable taking values in $[n]$ with $\Pb(\gamma = j) =  l_j/(\sum \limits_{k=1}^{n}l_k)$. Let $\xi$ be a Rademacher random variable, $\Pb(\xi=-1)=\Pb(\xi=1)=\frac{1}{2}$. Then,
\begin{align*}Z_{i}=\xi\sum \limits_{j \in [n]} \mathbbm{1}_{\{j\}}(\gamma)\frac{1}{\sqrt{\beta_1  l_j}}E_{i,j}
\end{align*}
where $E_{i,j}$ is an $m \times n$ matrix with $1$ in the $(i,j)^{th}$ entry and $0$ everywhere else.

\end{definition}

\subsection{Proof of Theorem \ref{t:less}}

These modifications allow us to prove subspace embedding guarantees for sparser matrices than in the oblivious case, thereby showing Theorem \ref{t:less}. In particular, we show that with LESS it suffices to use $O(\log^4 (d/\delta ))$ nonzero entries per row of $S$, instead of per column of $S$, which is much sparser since $S$ is a wide matrix.

\begin{theorem}\label{nonose}
Let $U$ be an arbitrary $n \times d$ deterministic matrix such that $U^TU=I$ with $(\beta_1,\beta_2)$-approximate leverage scores $(l_1,...,l_n)$. There exist constants $c_{\ref{nonose}.1},c_{\ref{nonose}.2},c_{\ref{nonose}.3}$, such that for any $0<\varepsilon<1, 0<\delta < 1$, and any LESS-IND-ENT or LESS-IND-ROWS random matrix $S$ corresponding to  $(l_1,...,l_n)$ with embedding dimension $m \geq c_{\ref{nonose}.1}  \max (d, \log(4/\delta))/\varepsilon^2$ and parameter $p \geq c_{\ref{nonose}.2} {(\log (d/\delta ))^4}/{( m \varepsilon ^6)}$, we have

\begin{align*}
\Pb \left( 1 - \varepsilon  \leq s_{\min}((1/\sqrt{pm})SU)   \leq s_{\max}((1/\sqrt{pm})SU) \leq 1 + \varepsilon \right) \geq 1-\delta, 
\end{align*}
and if we choose $m = c_{\ref{nonose}.1}  \max (d, \log(4/\delta))/\varepsilon^2$, $p = c_{\ref{nonose}.2} {(\log (d/\delta ))^4}/{( m \varepsilon ^6)} $, we have the following high probability bound for the maximum number of nonzero entries per row
\begin{align*}
\Pb \left( \mathop {\max }\limits_{i \in [m]} (\card(\{j \in [n]:{s_{ij}} \ne 0\} )) \le c_{\ref{nonose}.3}\beta_1  \beta_2{(\log (d/\delta ))^4}/{ \varepsilon ^4} \right) \geq 1-\delta
\end{align*}
Alternatively, given a fixed $\theta<3$, there exist constants $c_{\ref{nonose}.4}$, $c_{\ref{nonose}.5}$, $c_{\ref{nonose}.6}$ and $c_{\ref{nonose}.7}$ such that for $m \geq \max \{ (1+\theta)d, c_{\ref{nonose}.4}\log(4/\delta))/\theta^2 \}$ and $p>c_{\ref{nonose}.5}{(\log (d/\delta ))^4}/{(m\theta^6)}$, 

\begin{align*}
    \Pb \left( \kappa(\frac{1}{\sqrt{pm}}SU) \leq \frac{c_{\ref{nonose}.6}}{\theta} \right) \geq 1-\delta
\end{align*}
and if we choose $m = \max \{ (1+\theta)d, c_{\ref{nonose}.4}\log(4/\delta))/\theta^2 \}$ and $p=c_{\ref{nonose}.5}{(\log (d/\delta ))^4}/{(m\theta^6)}$, we have the following high probability bound for the maximum number of nonzero entries per row
\begin{align*}
\Pb \left( \mathop {\max }\limits_{i \in [m]} (\card(\{j \in [n]:{s_{ij}} \ne 0\} )) \le c_{\ref{nonose}.7}\beta_1  \beta_2{(\log (d/\delta ))^4}/{ \theta ^6} \right) \geq 1-\delta
\end{align*}

\end{theorem}
To prove Theorem \ref{nonose}, we need to first calculate the covariance profile of the three models.

\begin{lemma}[Variance and Uncorrelatedness] \label{lem:entriesnonobliv} 
Let $S=\{s_{ij}\}_{i \in [m], j \in [n]}$ be a $m \times n$ random matrix with LESS-IND-ENT or LESS-IND-ROWS distribution with parameter $p$. Then, $\E(s_{ij})=0$ and $\operatorname{Var}(s_{ij})=p$ for all $i \in [m], j \in [n]$, and $\cov(s_{i_1j_1},s_{i_2j_2})=0$ for any $\{i_1,i_2\} \subset [m], \{j_1,j_2\} \subset [n]$ and $ (i_1,j_1)\neq (i_2,j_2) $
\end{lemma}
\begin{proof}
For the independent entries case, we have $\Var(s_{ij})=\frac{1}{\beta l_j}\beta l_j p=p$ and the entries are not correlated because they are independent.

For the independent rows case, since $S$ is a sum of i.i.d. copies of $Z_{i'}$, we analyze $Z_{i'}$ first for some fixed $i'$. By definition, we have 
\begin{align*}Z_{i'}=\xi(\sum \limits_{j \in [n]} \mathbbm{1}_{\{j\}}(\gamma)\frac{1}{\sqrt{\beta_1  l_j}}E_{i',j})
\end{align*}
Let $z_{ij}$ be the $(i,j)\textsuperscript{th}$ entry of $Z_{i'}$. Then, we have
\begin{align*}z_{ij}=\frac{1}{\sqrt{\beta_1  l_j}}\mathbbm{1}_{\{(i,j)\}}(i', \gamma)\xi\end{align*}
Since $\xi$ and $\gamma$ are independent, we know that $\E(z_{ij})=\frac{1}{\sqrt{\beta_1  l_j}}\E(\mathbbm{1}_{\{(i,j)\}}(i', \gamma))\E(\xi)=0$ because $\E(\xi)=0$. Also, we have \begin{align*}\E(z_{ij}^2)=(\frac{1}{\sqrt{\beta_1  l_j}})^2\E(\mathbbm{1}_{\{(i,j)\}}(i',\gamma)^2)\E(\xi^2)=\frac{\mathbbm{1}_{\{ i \}}(i')}{\beta_1  (\sum \limits_{j=1}^{n}l_j)}\end{align*} 
Now, for the covariance, since the random variables $z_{ij}$ are centered, we have
\begin{align*}
\cov(z_{i_1j_1},z_{i_2j_2})&=\E(z_{i_1j_1} z_{i_2j_2}) \\ 
&=\E(\frac{1}{\sqrt{\beta_1  l_{j_1}}}\mathbbm{1}_{\{(i_1,j_1)\}}(i', \gamma)\xi \frac{1}{\sqrt{\beta_1  l_{j_2}}} \mathbbm{1}_{\{(i_2,j_2)\}}(i',\gamma)\xi)\\
&=\frac{1}{\sqrt{\beta_1  l_{j_1}}}\frac{1}{\sqrt{\beta_1  l_{j_2}}}\E(\xi^2)\E(\mathbbm{1}_{\{(i_1,j_1)\}}(i',\gamma)\mathbbm{1}_{\{(i_2,j_2)\}}(i',\gamma))
\end{align*} 
Now, if $(i_1,j_1) \ne (i_2,j_2)$, and $i_1 \neq i_2$, then $\mathbbm{1}_{\{(i_1,j_1)\}}(i',\gamma)\mathbbm{1}_{\{(i_2,j_2)\}}(i',\gamma)=0$ since $i'$ can equal only one of $i_1$ or $i_2$. If $i_1=i_2$, then $j_1 \neq j_2$ and $\mathbbm{1}_{\{(i_1,j_1)\}}(i',\gamma)\mathbbm{1}_{\{(i_2,j_2)\}}(i',\gamma)=0$ because $\gamma$ can only take one value out of $j_1$ and $j_2$.

Now, since \begin{align*}S= \sum_{i'=1}^m \sum \limits_{k=1}^{(\beta_1  p  \sum l_i)} Z_{i'k} \text{,}\end{align*} we have, by independence, $\Var(s_{ij}) = \sum_{i'=1}^m (\beta_1  p  \sum l_i) \frac{\mathbbm{1}_{\{ i \}}(i')}{\beta_1  (\sum \limits_{j=1}^{n}l_j)}=p$ and, when $(i_1, j_1)\neq (i_2, j_2)$,
\begin{align*}
    \cov(s_{i_1 j_1}, s_{i_2 j_2}) &= \E(s_{i_1 j_1}s_{i_2 j_2}) = \sum_{i'=1}^m \sum_{k=1}^{(\beta_1  p  \sum l_i)} \E \left( (Z_{i'k})_{i_1 j_1}(Z_{i'k})_{i_2 j_2} \right) = 0
\end{align*}

\end{proof}

Now, we can prove Theorem \ref{nonose}.

\begin{proof}[Proof of Theorem \ref{nonose}]
By Lemma \ref{lem:entriesnonobliv}, we see that both distributions satisfy the condition of Lemma \ref{mapa}, so we conclude that $\sigma_*(X) \leq 2\sqrt{p}$ and $\sigma(\augsym(SU)) \leq \sqrt{pm}$.

For the IND-ENT case, since $|s_{i,j}|$ is bounded by $\frac{1}{\sqrt{\beta_1 l_j}}$, we have
\begin{align*}R(\augsym(SU)) \le \max \limits _{i,j} \frac{1}{\sqrt{\beta_1 l_j}}  \left\| {\left[ {\begin{array}{*{20}{c}}
          0&0&{{{({e_i}{u_j}^T)}^T}}&0 \\ 
          0&0&0&0 \\ 
          {({e_i}{u_j}^T)}&0&0&0 \\ 
          0&0&0&0 
        \end{array}} \right]} \right\|_{op}  \le 1
\end{align*}

And for the IND-ROWS case, we have $Z_kU=\sum \limits _{i \in [m], j \in [n]}(Z_k)_{i,j}({e_i}{u_j}^T)$. Since this sum has only one nonzero term, we have $\norm{Z_kU}_{op} = \max \limits _{i \in [m], j \in [n]} |(Z_k)_{i,j}| \norm{{e_i}{u_j}^T}_{op}$. Using this decomposition to $\augsym(SU)$ and the fact that  $|(Z_k)_{i,j}| \le \frac{1}{\sqrt{\beta_1 l_j}}$, we conclude
\begin{align*}R(\augsym(SU)) \le \max \limits _{i \in [m], j \in [n], k \in [\beta_1  p m \sum l_i]} (Z_k)_{i,j} \left\| {\left[ {\begin{array}{*{20}{c}}
          0&0&{{{({e_i}{u_j}^T)}^T}}&0 \\ 
          0&0&0&0 \\ 
          {({e_i}{u_j}^T)}&0&0&0 \\ 
          0&0&0&0 
        \end{array}} \right]} \right\|_{op} \le 1
\end{align*}

For both cases, since their matrix parameters $\sigma_*(X), \sigma(X)$ and $R(X)$ are the same as in Theorem \ref{osngeneral}, following the proof for Theorem \ref{osngeneral}, we conclude that
there exist constants $c_{1},c_{2},$, such that for any $\varepsilon, \delta > 0$, we have
\begin{align*}
\Pb \left( 1 - \varepsilon  \leq s_{\min}((1/\sqrt{pm})SU)   \leq s_{\max}((1/\sqrt{pm})SU) \leq 1 + \varepsilon \right) \geq 1-\delta
\end{align*}
when $m \geq c_{1}  \max (d, \log(4/\delta))/\varepsilon^2$ and $pm>c_{2}{(\log (d/\delta ))^4}/{\varepsilon ^6}$. The proof of the claim when $m \geq (1+\theta)d$ also follows in the same manner as Theorem \ref{osngeneral}.

Since $p$ is just a parameter in the LESS models, we want to know what the requirement $pm>c_{2}{(\log (d/\delta ))^4}/{\varepsilon ^6}$ (similarly $pm>c_{2}{(\log (d/\delta ))^4}/{\theta ^6}$) means for the average number of nonzero entries in each row. First, the average number of nonzero entries in each row will be  $\beta_1 p \sum_{j=1}^n l_j \leq \beta_1 \beta_2 pd$ for both cases, just by the construction of these matrices. Then we observe that the condition $pm>c_{2}{(\log (d/\delta ))^4}/{\varepsilon ^6}$ is equivalent to
\begin{align*}pd>c_{2}({(\log (d/\delta ))^4}/{\varepsilon ^6})(\frac{d}{m})\end{align*}

Since $m \ge d/\varepsilon^2$, we have $\frac{d}{m} \le \varepsilon^2$. Therefore, to meet the requirement \begin{align*}
    pd>c_{2}({(\log (d/\delta ))^4}/{\varepsilon ^6})(\frac{d}{m})
\end{align*}
it suffices to have $pd>c_{2}{(\log (d/\delta ))^4}/{\varepsilon ^4}$. So the optimal choice of $p$ leads to \begin{align*}
    \beta_1 \beta_2 c_{2}{(\log (d/\delta ))^4}/{\varepsilon ^4}
\end{align*} many nonzero entries in each row on average. 

In the case when $m \geq (1+\theta)d$ and $pm>c_{2}{(\log (d/\delta ))^4}/{\theta ^6}$, since we can only claim $d/m < 1$ in general, we need $pd>c_{2}{(\log (d/\delta ))^4}/{\theta ^6}$ so the average number of nonzero entries in each row would be $\beta_1\beta_2c_{2}{(\log (d/\delta ))^4}/{\theta ^6}$. 

We now obtain the high probability bound for the number of nonzero entries in each row.

For the independent entries distribution, the number of nonzero entries is
\begin{align*}
    \sum \limits_{j=1}^n \mathbbm{1}_{\R \backslash \{0\}} (s_{ij})=\sum \limits_{j=1}^n \delta_{ij}
\end{align*}
Therefore, the total variance of the sum is $v_1=\sum \limits_{j=1}^n \beta_1 l_jp(1-\beta_1l_jp) \le \beta_1  \beta_2  pd$.

By Bernstein's Inequality and using the fact that $\E(\sum \limits_{j=1}^n \delta_{ij})=\beta_1 p \sum_{j=1}^n l_j \le \beta_1 \beta_2 pd$, we have
\begin{align*}
    &\Pb(\sum \limits_{j=1}^n \delta_{ij} > (t+1) \beta_1  \beta_2  pd)
    \\ \le & \Pb(\sum \limits_{j=1}^n \delta_{ij} > \beta_1 p \sum_{j=1}^n l_j + t \beta_1  \beta_2  pd)
    \\\le& \exp(-\frac{(t \beta_1  \beta_2  pd)^2}{2(v_1+(t \beta_1  \beta_2  pd)/3)}) \\\le& \exp(-\frac{t \beta_1  \beta_2  pd}{4})
\end{align*}
for $t \ge 3$.

By union bound, we have
\begin{align*}
&\Pb \left( \mathop {\max }\limits_{i \in [m]} (\card(\{j \in [n]:{s_{ij}} \ne 0\} )) \ge (t+1) \beta_1  \beta_2  pd \right) \\\le& m\exp(-\frac{t \beta_1  \beta_2  pd}{4})
\\ =& \exp(-\frac{t \beta_1  \beta_2  pd}{4}+\log(c_1\frac{d}{\varepsilon^2}))
\end{align*}

Since $pd>c_{2}{(\log (d/\delta ))^4}/{\varepsilon ^4}$, we have
\begin{align*}
&\exp(-\frac{t \beta_1  \beta_2  pd}{4}+\log(c_1\frac{d}{\varepsilon^2})) \\\le& \exp(-\frac{t \beta_1  \beta_2  (c_{2}{(\log (d/\delta ))^4}/{\varepsilon ^4})}{4}+\log(c_1\frac{d}{\varepsilon^2}))
\end{align*}

When $t$ is large enough, we have $(1/2)\frac{t \beta_1  \beta_2  (c_{2}{(\log (d/\delta ))^4}/{\varepsilon ^4})}{4}>\log(c_1\frac{d}{\varepsilon^2})$, because $\frac{1}{\varepsilon^4}$ dominates $\log(\frac{1}{\varepsilon^2})$ and $(\log (d/\delta ))^4$ dominates $\log(d)$.

In that case, we have
\begin{align*}
    &\Pb \left( \mathop {\max }\limits_{i \in [m]} (\card(\{j \in [n]:{s_{ij}} \ne 0\} )) \ge (t+1) \beta_1  \beta_2  pd \right) \\\le&  \exp(-(1/2)\frac{t \beta_1  \beta_2  (c_{2}{(\log (d/\delta ))^4}/{\varepsilon ^4})}{4})
    \\\le&  \exp(-(1/2)\frac{t  (c_{2}{(\log (d/\delta ))^4})}{4})
\end{align*}
because $\beta_1 \beta_2 >1$ and $\varepsilon<1$. We can argue similarly for the $m=(1+\theta)d$ case.

We observed that $(\log (d/\delta ))^4> \log(1/\delta)$ since $d>1$. Therefore, as long as we choose $t$ such that $\frac{c_2t}{8}>1$, we will have
\begin{align*}
    \exp(-t \frac{c_2}{8}(\log (d/\delta ))^4) \le \exp(-(\log (d/\delta ))^4) \le \exp(-(\log (1/\delta ))) =\delta
\end{align*}

For the independent rows distribution, the number of nonzero entries in row $i$ is less than or equal to $\beta_1p\sum l_i \le \beta_1\beta_2 pd$, by construction.

\end{proof}

As an immediate corollary of Theorem \ref{nonose} we give a new subspace embedding guarantee for the Fast Johnson-Lindenstrauss Transform (FJLT). Recall that an FJLT preconditions the matrix $U$ with the Randomized Hadamard Transform (see Definition \ref{def:RHT}) to transform it into another matrix $V$ whose row norms can be well controlled. In this way, we obtain  approximate leverage scores for $V$ by construction rather than by estimation. Then, we can apply LESS-IND-ENT or LESS-IND-ROWS random matrices to the preconditioned matrix $V$ according to these approximate leverage scores.

\begin{theorem}[Analysis of Preconditioned Sparse OSE]\label{usrht}
Let $U$ be an arbitrary $n \times d$ deterministic matrix such that $U^TU=I$. There exist constants $c_{\ref{usrht}.1},c_{\ref{usrht}.2},c_{\ref{usrht}.3}$, such that for any $0<\varepsilon<1, \frac{2n}{e^d}<\delta < 1$, the following holds. Let $S=\Phi (\frac{1}{\sqrt{n}}) HD$ where $H$ and $D$ are as in definition \ref{hada} and \ref{def:RHT}, and $\Phi$ has LESS-IND-ENT or LESS-IND-ROWS distribution corresponding to uniform leverage scores $(d/n,...,d/n)$ with embedding dimension $m \geq c_{\ref{usrht}.1}  \max (d, \log(8/\delta))/\varepsilon^2$ and average number of nonzero entries per row $\ge c_{\ref{usrht}.2} {(\log (2d/\delta ))^4}/{ \varepsilon ^4} $.  Then we have 
\begin{align*}
\Pb \left( 1 - \varepsilon  \leq s_{\min}((1/\sqrt{pm})SU)   \leq s_{\max}((1/\sqrt{pm})SU) \leq 1 + \varepsilon  \right) \geq 1-\delta
\end{align*}
\end{theorem}

\begin{proof}[Proof of Theorem \ref{usrht}]
First, note that $\frac{1}{\sqrt{n}}HDU$ is an $n \times d$ matrix with orthonormal columns. Let $\mathcal{E}$ denote the event that the tuple $(l_1=d/n, l_2=d/n, \ldots, l_n=d/n)$ of numbers are $(16,1)$-approximate leverage scores for the matrix $\frac{1}{\sqrt{n}}HDU$. 
Clearly, $\sum_{i=1}^n l_i \leq d$. Since $2n \leq \delta e^d$, we have $\log(2n / \delta) \leq d$, and the claim from Lemma \ref{rhtrownorms} reads,
    \begin{align*}
        \Pb \left( \max \limits_{j=1,...,n} \norm{e_j^T(\frac{1}{\sqrt{n}}HDU)}_{op} \geq \sqrt{\frac{d}{n}}+\sqrt{\frac{8d}{n}} \right) \leq \frac{\delta}{2}
    \end{align*}
    Thus, with probability greater than $1-\delta/2$, we have, for all $j \in [n]$, 
    \begin{align*}
        \norm{e_j^T(\frac{1}{\sqrt{n}}HDU)}^2 < (1+2\sqrt{2})^2\left(\frac{d}{n}\right) < \frac{16d}{n} = 16l_j
    \end{align*}
    Thus, the conditions for $(d/n, d/n, \ldots, d/n)$ to be $(16,1)$-approximate leverage scores for the matrix $\frac{1}{\sqrt{n}}HDU$ are satisfied with probability greater than $1-\delta/2$, i.e.,
    \begin{align*}
        \Pb(\mathcal{E}) \geq 1 - \delta/2
    \end{align*}

Let $V=\frac{1}{\sqrt{n}}HDU$. Conditioning on the random matrix $V$, we can apply Theorem \ref{nonose} and obtain
\begin{align*}
\Pb \left( 1 - \varepsilon  \leq s_{\min}((1/\sqrt{pm})\Phi V)   \leq s_{\max}((1/\sqrt{pm})\Phi V) \leq 1 + \varepsilon  \right|V)(\omega) \geq 1-\delta/2
\end{align*}
on the event $\mathcal{E}$.

Therefore, by independence between $V$ and $\Phi$, we have
\begin{align*}
\Pb \left( 1 - \varepsilon  \leq s_{\min}((1/\sqrt{pm})SU)   \leq s_{\max}((1/\sqrt{pm})SU) \leq 1 + \varepsilon  \right) \geq (1-\delta/2)^2 \ge 1- \delta
\end{align*}
as long as $m \geq c_{\ref{nonose}.1}  \max (d, \log(8/\delta))/\varepsilon^2$ and average number of nonzero entries per row $\ge c_{\ref{nonose}.2} {(\log (2d/\delta ))^4}/{ \varepsilon ^4} $.

\end{proof}

%% file: applications.tex
\section{Analysis of Fast Subspace Embeddings}
\label{s:fast-analysis}
In this section, we give proofs for the fast subspace embedding results.
\subsection{Proof of Theorem \ref{t:fast-ose}}
  Recall that our goal is to construct an $m\times n$ OSE matrix $S$ with
  $m=O(d)$ that can be applied to an $A\in\R^{n\times d}$ in time
  $$O(\gamma^{-1}\nnz(A) + d^{2+\gamma}\polylog(d)).$$
    We will present two separate constructions for achieving this: 
    \begin{enumerate}
        \item Our first construction optimizes the time complexity of computing $SA$, including the dependence on polylogarithmic factors.
        \item Our second construction optimizes the number of uniform random bits required to generate $S$, trading off the dependence on polylogarithmic factors and on $\gamma$.
    \end{enumerate}
    \paragraph{\emph{1. Optimizing time complexity}} We combine several different transformations to construct the embedding $S$. Namely, we compute $SA$ by chaining four transformations, $SA=S_3(HD(S_2(S_1A)))$, defined as follows:
    \begin{enumerate}
        \item $S_1$ is an $m_1\times n$ OSNAP matrix with $m_1=O(d^{1+\gamma}\log d)$ rows and $O(1/\gamma)$ non-zeros per column. The cost of computing $S_1A$ is $O(\gamma^{-1}\nnz(A))$.
        \item $S_2$ is an $m_2\times m_1$ OSNAP matrix with $m_2=O(d\log d)$ rows and $O(\log d)$ non-zeros per column. The cost of computing $S_2(S_1A)$ is $O(m_1 d\log d)=O(d^{1+\gamma}\log^2d)$.
        \item $HD$ is an $m_2\times m_2$ randomized Hadamard transform (padded with zeros, if necessary). The cost of computing $HD(S_2S_1A)$ is $O(m_2d\log d)=O(d^2\log^2d)$. 
        \item $S_3$ is our Leverage Score Sparsified embedding from Theorem~\ref{t:less} having $m_3=O(d)$ rows, implemented with all $l_i=m_2/d$ and $O(\log^4d)$ non-zeros per row. The cost of computing $S_3(HDS_2S_1A)$ is $O(m_3d\log^4d)=O(d^2\log^4 d)$.
    \end{enumerate}
    Note that, via Lemma \ref{l:existing}, both $S_1$ and $S_2$ are constant-distortion subspace embeddings. Moreover, randomized Hadamard transform $HD$ ensures that with high probability the leverage scores of $HDS_2S_1A$ are nearly uniform. This allows us to use the leverage score sparsified embedding from Theorem \ref{t:less} in an oblivious way. Overall, the cost of computing $SA$ becomes $O(\gamma^{-1}\nnz(A)+d^{2+\gamma}\log^2 d + d^2\log^4 d)$. We note that, using standard choices for $\gamma$, this complexity can be simplified in the following ways: 
    \begin{itemize}
    \item  $O(\nnz(A) +
      d^{2+\theta})$ for any absolute constant $\theta > 0$, obtained by choosing $\gamma = \theta/2$;
    \item $O(\nnz(A)\log d +
      d^2\log^4 d)$, obtained by choosing $\gamma = 1/\log d$.
    \end{itemize}

    \paragraph{\emph{2. Optimizing random bits}}
    We next present a slightly modified construction that addresses the claim discussed in Remark \ref{r:random-bits}. Namely, we construct a constant-distortion oblivious subspace embedding $S$ that can be generated using only 
$\polylog(n)$ uniform random bits. To do that, we combine three embeddings, $S=S_3S_2S_1$, where:
  \begin{enumerate}
    \item $S_1$ is an OSNAP matrix with $m_1=O(d^{1+\gamma}\log^8(d))$ rows and $O(1/\gamma^3)$
      non-zero entries per column, which requires $O(\log d\log n)$ random bits;
    \item $S_2$ is an OSNAP matrix with $m_2=O(d\,\log^8(d))$ rows and
      $O(\log^3(d))$ non-zero entries per column, which also requires $O(\log d \log n)$ random bits;
    \item $S_3$ is the diagonal construction from Theorem \ref{t:direct} with
      $m_3=O(d)$ rows and $O(\log^4(d))$ non-zeros per column. This requires $O(m_2/m_3\cdot \log^4 d\log n  ) = O(\log^{12} d\log n)$ many random bits.
    \end{enumerate}
    \vspace{3mm}
    The overall cost of computing $SA$ is thus $O(\gamma^{-3}\nnz(A) + d^{2+\gamma}\polylog (d))$, which for any constant $\gamma>0$ matches our optimal construction up to the $\polylog(d)$ term, while requiring only $O(\log^{12} d \log n)$ uniform random bits. 
    
    \subsection{Proof of Theorem \ref{t:fast-epsilon}} Here, we provide the analysis for our fast low-distortion subspace embedding, i.e., with $\epsilon=o(1)$.
   First, we assume that $\nnz(A)\geq n$, by ignoring the rows of $A$ that are
   all zeros. Next, note that we can also assume without loss of generality that
    $d\geq  n^{c}$ for $c=0.1$. To see this, suppose otherwise, i.e., that
    $d<n^{c}$. Then, we can show that:
    \begin{align}\label{eq:wlog}
      O(\nnz(A) + d^4/\epsilon^2) = O(\gamma^{-1}\nnz(A) + d^\omega + \epsilon^{-6}d^{2+\gamma}),
    \end{align}
 where the left-hand side is the cost of applying a CountSketch matrix
$S$ (i.e., OSNAP with one non-zero per column) with embedding dimension $m=O(d^2/\epsilon^2)$, and then
performing the thin eigendecomposition $SA=UDV^\top$ to get
a size $d\times d$ sketch $DV^\top$ that is the desired subspace
embedding. To show \eqref{eq:wlog}, 
we consider two cases: 
\begin{enumerate}
\item Suppose that $\epsilon\geq n^{-1/4}$. Then, we have
  $d^4/\epsilon^2\leq n\leq\nnz(A)$, which implies \eqref{eq:wlog}.
\item Now, suppose otherwise, i.e., $\epsilon<n^{-1/4}$. Then, we get $\epsilon^{-6}d^2\geq
  \epsilon^{-2}nd^2\geq \epsilon^{-2}d^4$, which also implies \eqref{eq:wlog}.
\end{enumerate}

For the rest of the proof, we assume that $d\geq n^c$. The basic idea of our low-distortion embedding is to first use the constant-distortion embedding (from Theorem~\ref{t:fast-ose}) to compute rough approximations of the leverage scores, and then use those leverage scores estimates as input to our Leverage Score Sparsified embedding (from Theorem \ref{t:less}). We start by quoting a result by \cite{chepurko2022near} who show how one can quickly approximate the row norms of matrix $AR$, given $A\in\R^{n\times d}$ and $R\in\R^{d\times d}$, without computing the matrix product $AR$ itself. 
\begin{lemma}[Modified Lemma 7.2 from \cite{chepurko2022near}]\label{l:lev}
Given $A\in\R^{n\times d}$ and $R\in\R^{d\times d}$, we can
compute in $O(\gamma^{-1}(\nnz(A)+d^2))$ time estimates $\tilde l_1,...,\tilde l_n$ such that with probability $\geq 0.9$:
\begin{align*}
n^{-\gamma}\|e_i^\top AR\|_2^2\leq \tilde l_i\leq O(\log(n))\|e_i^\top
AR\|_2^2\quad\forall i\qquad\text{and}\qquad \sum_i\tilde l_i \leq O(1)\cdot \|AR\|_F^2.
\end{align*}
\end{lemma}
The idea of the construction of \cite{chepurko2022near} essentially follows the usual JL-type row-norm approximation strategy, e.g., from \cite{drineas2012fast}. But, while the standard approach is to compute constant-factor approximations in $O((\nnz(A)+d^2)\log n)$ time, they observe that we can obtain coarser $O(n^\gamma)$ factor approximations in true input sparsity time. The procedure is as follows. Let $G$ be a $d\times k$ Gaussian matrix, scaled by $1/\sqrt k$, where $k=O(\gamma^{-1})$. Then, we compute:
$$\tilde l_i=\|e_i^\top ARG\|_2^2.$$
Note that these can be computed in
time: $O(\gamma^{-1}d^2)$ for computing $RG$, plus $O(\gamma^{-1}\nnz(A))$
for computing $A(RG)$ and its row norms. These estimates will then with high probability satisfy the guarantee from lemma \ref{l:lev}. Finally, while this was not part of Lemma 7.2 in \cite{chepurko2022near}, we also observe that this construction satisfies an $O(1)$ upper bound, sharper than $O(\log n)$, if we consider the aggregate of all estimates, instead of the individual $\tilde l_i$'s. This holds since $\sum_i\tilde l_i = \|ARG\|_F^2$, which is an unbiased estimate of $\|AR\|_F^2$ so a constant factor approximation holds with probability $0.9$ via Markov's inequality (which can easily be improved to a high-probability bound, e.g., via Bernstein's inequality). We note that in all of our applications, an element-wise lower-bound and an aggregate upper-bound are sufficient.

Now, we are ready to describe our own procedure for constructing a low-distortion subspace embedding which is described in Algorithm \ref{alg:fast-epsilon}.

\begin{algorithm}
\begin{algorithmic}[1]
    \STATE Compute $S_1A$ with $S_1$ from Theorem~\ref{t:fast-ose}
    \STATE Compute $R$ such that  $S_1A =QR^{-1}$ for orthogonal $Q$
    \STATE Compute leverage score estimates
    $\tilde l_i\approx\|e_i^\top AR\|_2^2$ using Lemma  \ref{l:lev}
    \STATE Compute $S$, the Leverage Score
    Sparsified embedding from Theorem~\ref{t:less} using $\tilde
    l_1,...,\tilde l_n$
    \RETURN $S$
  \end{algorithmic}
  \caption{Fast low-distortion subspace embedding.}\label{alg:fast-epsilon}
\end{algorithm}
Note that thanks to the subspace
embedding property of $S_1$, we have:
\begin{align*}
 \|e_i^\top
  AR\|_2^2\approx_{O(1)}\|e_i^\top U\|_2^2=\ell_i(A),
\end{align*}
where $U$ is the
orthonormal basis matrix for $A$'s column span.
Instead of computing $\|e_i^\top
  AR\|_2^2$ directly, we approximate them using Lemma~\ref{l:lev}, to
  obtain $O(n^\gamma)$-approximations $\tilde l_i$. Having computed the $\tilde l_i$'s, we compute the leverage score
sparsified embedding $S$ from Theorem \ref{t:less}, initialized with $\tilde l_i$'s.   Setting $m=O(d/\epsilon^2)$ and using
  $O(n^\gamma\log^4(d)/\epsilon^4)$ non-zeros per row, we
  obtain a subspace embedding with distortion $\epsilon$. The cost
  of computing $SA$ is
  $$O(mdn^\gamma\log^4(d)/\epsilon^4)=O(n^\gamma
  d^2\log^4(d)/\epsilon^6).$$
  Since we assumed that $d\geq n^{c}$, by adjusting $\gamma$ we can
  replace $n^\gamma$ with $d^\gamma$. Also, note that computing $R$ takes $O(d^\omega)$ time and the cost of computing $S_1A$ and $\tilde l_i$'s is bounded by $O(\gamma^{-1}\nnz(A)+d^{2+\gamma}\log^4 d)$. Thus, the overall time complexity of Algorithm \ref{alg:fast-epsilon} is:
  \begin{align*}
      &O(\gamma^{-1}\nnz(A) + d^\omega + \epsilon^{-6}d^{2+\gamma}\log^4(d)).
  \end{align*}

\section{Analysis of Fast Linear Regression}
\label{s:regression-analysis}
In this section, we give proofs for the linear regression results.

\subsection{Proof of Theorem \ref{t:reduction}}
  This result follows immediately from our fast low-distortion embedding result (Theorem \ref{t:fast-epsilon}). Compute $[\tilde A\mid\tilde b]=S[A\mid b]$ where $S$ is our low-distortion
  embedding for $[A\mid b]$. Now, suppose that
  we find $\tilde x$ such that:
  \begin{align*}
    \tilde f(\tilde x)\leq (1+\epsilon)\min_{x\in \mathcal C} \tilde
    f(x),\quad\text{where}\quad \tilde f(x)=\|\tilde Ax-\tilde b\|_2^2+g(x).    
  \end{align*}
  Note that, the subspace embedding property implies that:
  \begin{align*}
    (1+\epsilon)^{-2}f(x)\leq \tilde f(x) \leq (1+\epsilon)^2 f(x),\quad
    \forall x\in \mathcal C.
  \end{align*}
  Then, defining $x^*=\argmin_{x\in\mathcal C} f(x)$, our estimate $\tilde x$ satisfies:
  \begin{align*}
    f(\tilde x)
    &\leq (1+\epsilon)^2 \tilde f(\tilde x)
    \leq (1+\epsilon)^4 \min_{x\in\mathcal C} \tilde f(x)
    \\
    &\leq (1+\epsilon)^4 \tilde f(x^*)\leq (1+\epsilon)^6f(x^*)\leq (1+10\epsilon)f(x^*)
  \end{align*}
  Adjusting the $\epsilon$ by a constant factor concludes the reduction.

\subsection{Proof of Theorem \ref{t:least-squares}}

Part 1 of Theorem \ref{t:least-squares} follows immediately from our fast oblivious subspace embedding result, Theorem \ref{t:fast-ose}, combined with the above reduction result, Theorem \ref{t:reduction}. Crucially, the algorithm relies on the fact that our embedding is oblivious, which allows us to construct the sketches  $SA$ and $Sb$ in a single pass over the data, and using the optimal $O(d^2\log(nd))$ bits of space.

Next, we move on to part 2 of Theorem \ref{t:least-squares}, i.e., computing a $1+\epsilon$ approximation of the least squares
solution, $x^* = \argmin_x\|Ax-b\|_2$ (not necessarily in a single pass).
We propose the following algorithm, which is a variant of 
preconditioned mini-batch stochastic gradient descent with $O(d^{1+\gamma})$ mini-batch size, using our
constant-distortion OSE (Theorem \ref{t:fast-ose}) for initialization and preconditioning. We note that this is a different approach than the one proposed by \cite{cherapanamjeri2023optimal}, who run preconditioned gradient descent on an $O(d\polylog(d)/\epsilon)$ sketch, resulting in an additional $O(\epsilon^{-1}\log(1/\epsilon)d^2\polylog(d))$ cost.
\begin{algorithm}
  \begin{algorithmic}[1]
    \STATE Compute $SA$ and $Sb$ with $S$ from Theorem~\ref{t:fast-ose}
    \STATE Compute $R$ such that  $SA =QR^{-1}$ for orthogonal $Q$
    \STATE Compute leverage score estimates
    $\tilde l_i$ using Lemma  \ref{l:lev}
    \STATE Compute $x_0 = \argmin_x\|SAx-Sb\|_2$
    \FOR{$t=0$ to $T-1$}
    \STATE Construct a sub-sampling matrix $S_t$ with $k$ rows using distribution
    $(p_1,...,p_n)$  where $p_i = \tilde l_i/\sum_i\tilde l_i$
    \STATE Compute $g_t \leftarrow 2(S_tA)^\top (S_tA x_t- S_tb)$
    \STATE Compute $x_{t+1}\leftarrow x_t - \eta_t RR^\top g_t$ 
    \ENDFOR
    \RETURN{$x_T$}
  \end{algorithmic}
  \caption{Fast least squares via preconditioned mini-batch SGD}
  \label{alg:least-squares} 
\end{algorithm}

To establish the result, we show the following convergence
guarantee for Algorithm \ref{alg:least-squares}. The result is potentially of independent interest, so we state it here in a general form using any mini-batch size $k$, and then later specify the right choice for our setting.
  \begin{lemma}\label{l:sgd}
    Given $A\in \R^{n\times d}$ and $b\in\R^n$, suppose that  the sub-sampling probabilities satisfy:
    \begin{align*}
      p_i\geq \ell_i(A)/(\alpha d)\qquad\forall_i.
    \end{align*}
    Then, there is an absolute constant $c>0$ such that, letting $\eta_t:= \frac{\beta}{1+\beta t/8}$ for $\beta =
    \frac{k/8}{k+4\alpha d}$, Algorithm \ref{alg:least-squares}
    satisfies:
    \begin{align*}
      \E[f(x_t) - f(x^*)] \leq \frac{f(x_0)}{1+kt/(c\alpha
      d)}\qquad\forall_{t\geq 1}. 
    \end{align*}
  \end{lemma}
  \begin{proof}
      Thanks to the subspace
  embedding property of $S$, we know that with high probability the
  condition number of $AR$ satisfies
  $\kappa(AR)^2\leq 2$. Consider the quadratic objective function $f(x) =
  \|Ax-b\|_2^2$. We start by showing a bound on   
  the variance of the stochastic gradient $g_t$ (which is an unbiased estimator of $\nabla f(x_t)$), as measured via $\E[\|M(g_t - \nabla f(x_t))\|_{\Sigma}^2]$. Here, we use
  the notation $\|x\|_\Sigma=\sqrt{x^\top \Sigma x}$, letting
  $\Sigma = A^\top A$ and $M = RR^\top$.

First, we use the fact that $R$ is a good preconditioner for
$A$:
\begin{align*}
  \E\big[\|M(g_t - \nabla f(x_t))\|_{\Sigma}^2\big]
  &=
  \E\big[\|\Sigma^{1/2} M\Sigma^{1/2}\Sigma^{-1/2}(g_t - \nabla
    f(x_t))\|_2^2
  \\
  &\leq 4 \,\E\big[\|\Sigma^{-1/2}(g_t-\nabla f(x_t))\|_2^2\big].
\end{align*}
Next, denoting $r_t = 2(Ax_t-b)$ and $U = A \Sigma^{-1/2}$, using the
approximation guarantee for the leverage score estimate, we obtain:
\begin{align*}
  \E\big[\|\Sigma^{-1/2}(g_t-\nabla f(x_t))\|_2^2\big]
  &=\E\big[\|U^\top S_t^\top S_t r_t - U^\top r_t\|_2^2\big]
  \\
  &=\frac1k \E\Big[\frac1{p_I^2}\|u_I(a_I^\top x_t-b_I)-U^\top
    r_t\big\|_2^2\Big]
  \\
  &\leq \frac1k \E\Big[\frac1{p_I^2}\|u_I\|_2^2(a_I^\top
    x_t-b_I)^2\Big]
  \\
  &\leq \frac1k\E\Big[\frac{\alpha d}{p_I}(a_I^\top x_t-b_I)^2\Big]
= \frac{\alpha d} k \|Ax_t-b\|_2^2.
\end{align*}
We next incorporate this bound on the stochastic gradient noise into
the SGD analysis to establish a one-step expected convergence
bound. Below, we use that $\E[g_t]=\nabla f(x_t) = 2A^\top(Ax_t-b) =
2A^\top A(x_t-x^*)$.
\begin{align*}
    \E\big[\|x_{t+1}-x^*\|_{\Sigma}^2\big]
  &= \|x_t-x^*\|_\Sigma^2- \eta_t(x_t-x^*)^\top\Sigma M\nabla f(x_t) +
    \frac{\eta_t^2}{2}\E\big[\|Mg_t\|_\Sigma^2\big]
  \\
  &= \|x_t-x^*\|_\Sigma^2- \eta_t\|x_t-x^*\|_{\Sigma M\Sigma}^2+
    \frac{\eta_t^2}{2}\E\big[\|Mg_t\|_\Sigma^2\big]
  \\
  &\leq \Big(1-\frac{\eta_t}{2}\Big) \|x_t-x^*\|_\Sigma^2 +
    \frac{\eta_t^2}{2}\E\big[\|Mg_t\|_\Sigma^2\big] .
\end{align*}
Next, we decompose the last term as follows:
\begin{align*}
  \E\big[\| Mg_t\|_\Sigma^2\big]
  &= \|M\nabla f(x_t)\|_\Sigma^2 +
    \E\big[\|M(g_t-\nabla f(x_t))\|_\Sigma^2\big]
  \\
  &\leq \|x_t-x^*\|_{\Sigma M\Sigma M\Sigma}^2 + \frac{4\alpha
    d}{k}f(x_t)
  \\
  &\leq 4\|x_t-x^*\|_\Sigma^2 + \frac{4\alpha d}{k}f(x_t).
\end{align*}
Putting this together, for $\eta_t\leq \min\{\frac14,\frac{k}{8\alpha
  d}\}$ we get:
\begin{align*}
  \E\big[f(x_{t+1})-f(x^*)\big] &\leq
  \Big(1-\eta_t+\frac{4\eta_t^2}{2}\Big)\big[f(x_t)-f(x^*)\big]+
  \frac{\eta_t^2}{2}\frac{4\alpha d}{k}f(x_t)
  \\
  &\leq \Big(1-\frac{\eta_t}{2}+\eta_t^2\frac{2\alpha
    d}{k}\Big)\big[f(x_t)-f(x^*)\big]+\eta_t^2\frac{2\alpha d}{k}f(x^*)
  \\
  &\leq \Big(1-\frac{\eta_t}{4}\Big) \big[f(x_t)-f(x^*)\big]+\eta_t^2\frac{2\alpha d}{k}f(x^*).
\end{align*}
Now, we proceed by induction. Let $\eta_t = \frac{\beta}{u_t}$, where
$u_t=1+\beta t/8$ and $\beta = \frac{k/4}{k+4\alpha d}$. Also, define
$G = 2\alpha d/k$. Our inductive hypothesis is:
\begin{align*}
  \E\big[f(x_t) - f(x^*)\big]\leq \frac{ f(x_0)}{u_t} =
  \frac{f(x_0)}{1+ \frac{k/32}{k+4\alpha d}\,t}.
\end{align*}
 The base case, $t=0$, is obvious since $u_0=1$. Now, suppose that the
 hypothesis holds for some $t$. Then, using that $4\beta G\leq 1/2$, we have:
 \begin{align*}
   \E\big[f(x_{t+1}) - f(x^*)\big]
   &\leq
   \Big(1-\frac{\beta}{4u_t}\Big)\E\big[f(x_t)-f(x^*)\big]
   +\frac{\beta^2}{u_t^2}G f(x^*)
   \\
   &\leq
     \Big(1-\frac{\beta}{4u_t}\Big)\frac{f(x_0)}{u_t}+\frac{\beta^2}{u_t^2}G
     f(x_0)
   \\
   &= \Big(1-\frac{\beta}{4u_t}(1-4\beta G)\Big) \frac{f(x_0)}{u_t}
   \\
   &\leq \Big(1-\frac{\beta}{8u_t}\Big)  \frac{f(x_0)}{u_t} =
     \frac{8u_t-\beta}{8u_t}\,\frac1{u_t}\,f(x_0)
   \\
   &\leq \frac{8}{8u_t+\beta} f(x_0) = \frac{f(x_0)}{u_t+\beta/8}=\frac{f(x_0)}{u_{t+1}}.
 \end{align*}
 This concludes the proof of Lemma \ref{l:sgd}.
  \end{proof}

 Now, to finalize the proof of Theorem \ref{t:least-squares} part 2, we set $\alpha
 = O(n^\gamma)$ and $k=\alpha d$, so that $\E[f(x_t)-f(x^*)]\leq
 f(x_0)/t$. Note that we initialized $x_0$ so that $f(x_0)=O(1)\cdot f(x^*)$. Thus, the complexity of the mini-batch SGD (excluding the preprocessing) is $O(kd/\epsilon) = O(n^\gamma d^{2}/\epsilon)$. Now, including the $O(\gamma^{-1}\nnz(A) + d^\omega)$ time for computing $R$ and the leverage score estimates, we obtain the overall complexity of:
 \begin{align*}
     O(\gamma^{-1}\nnz(A) + d^\omega + n^\gamma d^2/\epsilon).
 \end{align*}
 Here, again, we can assume without loss of generality that
$\nnz(A)\geq n$ and that $d\geq n^c$ for $c=0.1$, which gives the final claim.